\documentclass{IEEEtran}

\usepackage{float}
\usepackage{graphicx}
\usepackage{clrscode}
\usepackage{stfloats}
\usepackage{paralist}
\usepackage{tikz}
\usepackage{soul} 
\usepackage{color}
\usepackage{cite}
\usepackage{multicol}
\usepackage{comment}

\usepackage{amsmath, amssymb, mathrsfs, amsfonts}
\usepackage{amsthm}
\usepackage{mathtools}
\usepackage{epsfig, epstopdf}

\usepackage{algorithm}
\usepackage{enumitem}
\usepackage{algorithmic}
\usepackage{lipsum}
\usepackage{pbox}
\usepackage{array}
\usepackage{tablefootnote}

\usepackage{subfigure}

\usetikzlibrary{arrows,automata}

\allowdisplaybreaks

\newtheorem{lemma}{Lemma}

\newtheorem{proposition}{Proposition}

\newtheorem{remark}{Remark}

\usepackage{hyperref}

\newcommand{\RNum}[1]{\uppercase\expandafter{\romannumeral #1\relax}}

\begin{document}

\title{Fundamental Limits of Random Downlink Integrated Sensing and Communication over Rician Channels}

\author{
		\IEEEauthorblockN{Marziyeh Soltani, Mahtab Mirmohseni, \textit{Senior Member, IEEE},  Rahim Tafazolli, \textit{Fellow, IEEE},\\ and Mark F. Flanagan, \textit{Senior Member, IEEE}\\
		\vspace*{0.5em}
			}\thanks{This work was supported by Taighde Éireann - Research Ireland under the Frontiers for the Future Programme (Grant Number 24/FFP-P/12895).\\
            Marziyeh Soltani and Mark F. Flanagan are with the School of Electrical and Electronic Engineering, University College Dublin, Belfield, Dublin 4, D04 V1W8, Ireland (e-mail: marziyeh.soltani@ucd.ie, mark.flanagan@ieee.org). \\ Mahtab Mirmohseni and Rahim Tafazolli are with the 6G Innovation Centre (6GIC), Institute for Communication Systems (ICS), University of Surrey, Guildford GU2 7XH, U.K. (e-mail: m.mirmohseni@surrey.ac.uk, r.tafazolli@surrey.ac.uk). }}

\maketitle
\begin{abstract}
This paper studies the stochastic performance of a downlink multiple-input multiple-output integrated sensing and communication (ISAC) system over Rician fading channels. Rician fading is important in line-of-sight (LoS)-dominated deployments, where a deterministic propagation component can strongly affect both sensing and communication reliability. The considered base station (BS) simultaneously serves a user and senses a target. The BS–user channel contains both LoS and non-line-of-sight (NLoS) components. The user LoS angle may be fixed or random, and the target angle may follow an arbitrary distribution that is potentially correlated with the user angle. Compared with the Rayleigh setting, the deterministic LoS component introduces coherent angle-dependent terms and makes the key random vectors in the Gaussian approximation generally independent but non-identically distributed, which requires new analysis.
We analyze two beamforming strategies: subspace joint beamforming (SJB), which is optimal for the shared sensing-communication waveform structure considered here, and linear beamforming (LB), a near-optimal and practically attractive alternative that uses separate sensing and communication beamformers. For both schemes, we derive the communication outage probability (OP) and the sensing OP based on the Cram\'er--Rao bound (CRB). We also identify several practically important special cases that admit simpler expressions. For LB, we derive upper and lower bounds on the sensing OP together with a tractable approximation.
We further characterize the large-system and high-power scaling laws. LB without dirty paper coding (DPC) is interference-limited in the high-power regime due to radar self-interference. The results show that the Rician $K$-factor affects communication more strongly than sensing, with non-monotonic behavior across propagation regimes. Under the considered settings, LB with DPC provides the best overall performance in strong-LoS environments and is the only scheme observed to achieve ultra-high communication reliability in the special case of Rayleigh fading,, while SJB offers a robust lower-complexity alternative across a broad range of operating conditions.
\end{abstract}
\section{Introduction}\label{introduction}
\IEEEPARstart{T}{he} growing need for new capabilities in wireless networks has led to the integration of different functions that were traditionally separate. A key example of this trend is integrated sensing and communications (ISAC), which brings together radar and communication systems to simultaneously perform sensing and data transmission \cite{SeventyYearsofRadarandCommunications}. This dual-functional approach is especially important for emerging applications in the internet of things (IoT), including smart city infrastructure, vehicular, and industrial networks \cite{Integratedtoward}. By sharing spectrum and hardware resources, ISAC offers significant advantages for both sensing and communication. However, this resource sharing creates an inherent trade-off between sensing and communication performance, making the characterization of their fundamental limits crucial for ISAC system design \cite{ASurveyonFundamentalLimits}.
\subsection{Related Works}
The fundamental limits of ISAC systems have been extensively studied from two main perspectives. The first adopts an information-theoretic approach, characterizing the capacity-distortion tradeoff for various channel models 
%\cite{ASurveyonFundamentalLimits,2018JointStateSensingandCommunicationOptimalTradeoffforaMemorylessCase,2019JointStateSensingandCommunicatiooverMemorylessMultipleAccessChannels,2020JointSensingandCommunicationoverMemorylessBroadcastChannels,2022AnInformation-TheoreticApproachtoJointSensingandCommunication}
\cite{ASurveyonFundamentalLimits,2018JointStateSensingandCommunicationOptimalTradeoffforaMemorylessCase,2022AnInformation-TheoreticApproachtoJointSensingandCommunication}. The second focuses on specific sensing design metrics, with the transmit beampattern \cite{TowardDualfunctionalRadarCommunicationSystems,MUMIMOCommunicationsWithMIMORadar,JointTransmitBeamformingforMultiuser,optimaltransmitbeamformingintegrated} and the Cramér–Rao bound (CRB) \cite{CrameRaoBoundOptimizationforJoint,fromtorchtoprojector,MIMOIntegratedSensingandCommunicationCRBRateTradeoff,fundamentalcrbratetradeoffmultiantenna} emerging as the most common performance measures for estimation tasks. While these works have laid the groundwork for ISAC system design, they predominantly assume deterministic channels. However, practical wireless environments are inherently random, necessitating stochastic performance analysis.

More recently, a growing body of literature has investigated the probabilistic behavior of ISAC systems by incorporating randomness into channel models to more accurately capture practical propagation environments \cite{OnthePerformanceofUplinkandDownlink,PerformanceAnalysisandPowerAllocationforCooperative,Aunifiedperformanceframeworkfor,NOMAISACPerformanceAnalysisandRateRegion,MIMOISACPerformanceAnalysis,PerformanceAnalysioftheFullDuplexJoint,PerformanceofDownlinkandUplinkIntegratedSensing,onthefundementaltradeoff,Network-levelISAC:AnAnalyticalStudyofAntenna,coverageandrateofjointcommunication,PerformanceAnalysisofCooperativeIntegratedSensing,networklevelintegratedsensingcommunication,PerformanceAnalysisofISACWithActiveMulti,OnStochasticFundamentalLimitsina,Fundamentalchannelcoupling}. 
%\cite{OnthePerformanceofUplinkandDownlink,PerformanceAnalysisandPowerAllocationforCooperative,Aunifiedperformanceframeworkfor,NOMAISACPerformanceAnalysisandRateRegion,MIMOISACPerformanceAnalysis,PerformanceAnalysioftheFullDuplexJoint,PerformanceofDownlinkandUplinkIntegratedSensing,onthefundementaltradeoff,cooperativeisacnetworksperformance,Network-levelISAC:AnAnalyticalStudyofAntenna,ISACNetworkPlanning:SensingCoverageAnalysis,coverageandrateofjointcommunication,PerformanceAnalysisofCooperativeIntegratedSensing,networklevelintegratedsensingcommunication,PerformanceAnalysisofISACWithActiveMulti,OnStochasticFundamentalLimitsina,Fundamentalchannelcoupling}.
These studies can be broadly categorized based on the sensing metric employed. The first adopts detection probability \cite{OnthePerformanceofUplinkandDownlink,PerformanceAnalysisandPowerAllocationforCooperative,Aunifiedperformanceframeworkfor}. The second focuses on mutual information or signal-to-noise ratio (SNR), deriving rates under Rayleigh \cite{NOMAISACPerformanceAnalysisandRateRegion,MIMOISACPerformanceAnalysis,PerformanceAnalysioftheFullDuplexJoint,PerformanceofDownlinkandUplinkIntegratedSensing} and Rician fading \cite{PerformanceAnalysisofISACWithActiveMulti}. The third, most relevant to this paper, employs the CRB for estimation tasks, investigating the communication rate and estimation accuracy tradeoff under general block fading \cite{onthefundementaltradeoff} as well as Rayleigh \cite{OnStochasticFundamentalLimitsina}, and Rician fading \cite{Fundamentalchannelcoupling}. 
Network-level ISAC studies have extended these analyses to multi-cell scenarios using stochastic geometry under Rayleigh fading assumptions \cite{Network-levelISAC:AnAnalyticalStudyofAntenna,coverageandrateofjointcommunication,PerformanceAnalysisofCooperativeIntegratedSensing,networklevelintegratedsensingcommunication}.
%\cite{cooperativeisacnetworksperformance,Network-levelISAC:AnAnalyticalStudyofAntenna,ISACNetworkPlanning:SensingCoverageAnalysis,coverageandrateofjointcommunication,PerformanceAnalysisofCooperativeIntegratedSensing,networklevelintegratedsensingcommunication}.
\subsection{Limitations of existing works}
While the works discussed above have advanced the stochastic analysis of ISAC, they nonetheless suffer from several important limitations. First, the vast majority of existing stochastic ISAC studies assume Rayleigh fading for the user channel \cite{OnthePerformanceofUplinkandDownlink,PerformanceAnalysisandPowerAllocationforCooperative,Aunifiedperformanceframeworkfor,NOMAISACPerformanceAnalysisandRateRegion,MIMOISACPerformanceAnalysis,PerformanceAnalysioftheFullDuplexJoint,PerformanceofDownlinkandUplinkIntegratedSensing,onthefundementaltradeoff,Network-levelISAC:AnAnalyticalStudyofAntenna,coverageandrateofjointcommunication,PerformanceAnalysisofCooperativeIntegratedSensing,networklevelintegratedsensingcommunication,OnStochasticFundamentalLimitsina}
%\cite{OnthePerformanceofUplinkandDownlink,PerformanceAnalysisandPowerAllocationforCooperative,Aunifiedperformanceframeworkfor,NOMAISACPerformanceAnalysisandRateRegion,MIMOISACPerformanceAnalysis,PerformanceAnalysioftheFullDuplexJoint,PerformanceofDownlinkandUplinkIntegratedSensing,onthefundementaltradeoff,cooperativeisacnetworksperformance,Network-levelISAC:AnAnalyticalStudyofAntenna,ISACNetworkPlanning:SensingCoverageAnalysis,coverageandrateofjointcommunication,PerformanceAnalysisofCooperativeIntegratedSensing,networklevelintegratedsensingcommunication,OnStochasticFundamentalLimitsina}
, which models rich scattering with no dominant line-of-sight (LoS) component. However, in many emerging wireless scenarios, the channel contains a strong LoS component. The Rician fading model accurately captures such environments, with the Rician K-factor quantifying the LoS-to-non-line-of-sight (NLoS) power ratio. Moreover, the LoS angle may be random (due to user mobility) or fixed (in static deployments), and may exhibit an arbitrary joint distribution with the target angle. Despite its practical relevance, Rician fading has received little attention in the ISAC literature.

Considering estimation as the sensing task (the focus of this paper), mutual information based studies \cite{NOMAISACPerformanceAnalysisandRateRegion,MIMOISACPerformanceAnalysis,PerformanceAnalysioftheFullDuplexJoint,PerformanceofDownlinkandUplinkIntegratedSensing,PerformanceAnalysisofISACWithActiveMulti} assume perfect knowledge of the target's transmit correlation matrix and do not explicitly model the angle of arrival (AoA). CRB-based studies also exhibit limitations: \cite{onthefundementaltradeoff} only examines two extreme points on the CRB–rate region, treats randomness in the communication channel only conceptually, using a block fading model but without specifying a particular fading distribution (neither Rayleigh nor Rician), and overlooks beamforming design. The Rician study in \cite{Fundamentalchannelcoupling} captures channel-coupling effects, but it assumes deterministic user LoS angles, treats target angles as fixed parameters rather than random variables, and does not characterize the communication-outage/CRB-outage region under arbitrary joint user-target angle distributions. Works on network-level ISAC \cite{Network-levelISAC:AnAnalyticalStudyofAntenna,PerformanceAnalysisofCooperativeIntegratedSensing,networklevelintegratedsensingcommunication}
%\cite{cooperativeisacnetworksperformance,Network-levelISAC:AnAnalyticalStudyofAntenna,ISACNetworkPlanning:SensingCoverageAnalysis,PerformanceAnalysisofCooperativeIntegratedSensing,networklevelintegratedsensingcommunication}
assume Rayleigh fading, employ zero-forcing beamforming to suppress sensing–communication interaction, treat beamforming gains as fixed constants, disregard fast fading in the sensing channel by replacing random fading coefficients with their expected values, and rely on non-coherent distributed multiple-input multiple-output (MIMO) radar with multi-base station (BS) coordination. Similarly, \cite{coverageandrateofjointcommunication} considers delay/Doppler estimation in mmWave networks under Nakagami fading, but does not model the target angle, uses simplified sectored beamforming, and provides only upper and lower bounds on coverage probability and ergodic rate
\subsection{Our Contributions}
In this paper, we address these limitations by characterizing ISAC performance under optimal and near-optimal beamforming, while explicitly accounting for randomness in both the user channel and the target angle. We consider a downlink MIMO ISAC system where a multi-antenna BS transmits a dual-functional signal to simultaneously serve a user and sense a target. In contrast to the work of \cite{OnStochasticFundamentalLimitsina,soltani2023outage} (which assumed Rayleigh fading and a uniform target angle), here we adopt a Rician BS–user channel where the LoS angle may be deterministic or random, and the target angle follows an arbitrary joint distribution with the user angle, allowing for arbitrary correlation. We use the joint outage performance region and its Pareto boundary as a convenient way to summarize the sensing--communication trade-off in this richer stochastic setting. 

The extension from Rayleigh to Rician fading is not a routine parameter change. In the Rayleigh case, the zero-mean i.i.d. channel structure yields i.i.d. summands in the Gaussian approximation and eliminates many mixed terms by symmetry. In the Rician case, however, the deterministic LoS component introduces coherent angle-dependent terms involving both the target and user directions, so the key random vectors become independent but generally non-identically distributed, and the outage expressions must be averaged over an arbitrary joint angle distribution.
The main contributions of this paper are as follows:
\\
$\bullet$ We study two beamforming strategies. The first is subspace joint beamforming (SJB), which is optimal for the shared sensing--communication waveform structure considered here, in the sense that for each channel realization the CRB-minimizing beamformer under a communication constraint lies in the span of the target steering vector and the user channel \cite{CrameRaoBoundOptimizationforJoint}. The second is linear beamforming (LB), which uses separate communication and sensing beamformers and fully exploits the MIMO radar degrees of freedom (DoF).
\\
$\bullet$ Under the Rician fading model, we analytically characterize the communication outage probability (OP) and the sensing OP for these schemes, where the sensing OP is defined as the probability that the CRB for the target angle exceeds a prescribed threshold. For SJB, the analysis yields tractable Gaussian formulations based on the multivariate central limit theorem (CLT). For LB, the communication OP is characterized analytically, while the sensing OP is studied through analytical upper and lower bounds together with a tractable approximation. This outage-based formulation is well suited to slow fading and real-time sensing applications requiring guaranteed estimation accuracy.
\\
$\bullet$ We further derive simplified results for several practically relevant special cases, including Rayleigh fading user channels, pure LoS user channels, aligned user-target directions, array-null user-target angular separations, and highly concentrated angle distributions. We also analyze opportunistic sensing and opportunistic communication operating points, which provide additional insight into the boundary behavior of the sensing--communication trade-off.
\\
$\bullet$ We characterize the large-system and high-power scaling laws. In particular, the CRB scaling law is more favorable for LB than for SJB, while LB without dirty paper coding (DPC)\cite{costa} is interference-limited in the high-power regime due to radar self-interference. Under the considered LB design and parameter regimes, DPC is highly beneficial for mitigating radar self-interference at the user; without it, communication reliability is severely limited, especially in the high-power regime. Our analysis and simulations show that the Rician $K$-factor affects communication more strongly than sensing, with regime-dependent and generally non-monotonic behavior across propagation conditions.
\\
The paper is organized as follows: Section \ref{systemmodel} presents the system model and beamforming methods. Sections \ref{jointbeamformingperformance} and \ref{linearbeamforing} analyze SJB and LB performance, respectively. Section \ref{specialpoints} analyzes opportunistic sensing and communication. Section \ref{asympt} derives the asymptotic scaling laws. Numerical results and conclusions are given in Sections \ref{simulations} and \ref{conclusion}.
\\
\textbf{Notation:} We use lowercase, boldface lowercase, and boldface uppercase letters to denote scalar quantities, vectors, and matrices, respectively. $P(\cdot), f_x(\cdot)$, and $E[\cdot]$ represent probability, probability density function (PDF), and expectation, respectively. $\mathbf{X}$, $\mathbf{X}^T, \mathbf{X}^H$, and $\mathbf{X}^*$ denote a matrix with $M$ rows and $N$ columns, its transpose, Hermitian transpose, and conjugate, respectively. The Euclidean norm of a vector is denoted by $\parallel \cdot \parallel$, and the magnitude of a complex number is denoted by $|\cdot|$. $\mathcal{CN}(\mu,\sigma^2)$ represents a circularly-symmetric complex Gaussian distribution with mean $\mu$ and variance $\sigma^2$, $\mathcal{N}(\boldsymbol{\mu}, \mathbf{\Sigma})$ is a multivariate Gaussian distribution with mean vector $\boldsymbol{\mu}$ and covariance matrix $\mathbf{\Sigma}$. $\mathbb{C}$ and $\mathbb{R}$ denote the set of complex and real numbers, respectively. $\mathrm{var}(\cdot)$ and $\mathrm{cov}(\cdot)$ denote variance and covariance, respectively. $\overset{d}\rightarrow$ and $\overset{p}\rightarrow$ indicate convergence in distribution and probability, respectively. \(\stackrel{\text{d}}{\approx}\) indicates approximation in distribution. \text{Tr($\mathbf{A}$)} represents the trace of matrix $\mathbf{A}$. \(\mathcal{R}\) and \(\mathcal{I}\) denote the real and imaginary parts of a complex number, respectively. \(a_N = \mathcal{O}(b_N)\) and \(X_N = \mathcal{O}_p(a_N)\) indicates that \(|a_N| \leq C|b_N|\) for some constant \(C\) and sufficiently large N, and that \(X_N/a_N\) is bounded in probability, respectively. \(\xrightarrow{\text{a.s.}}\) indicates almost sure convergence. $a \propto b$ indicates that $a = c \cdot b$ for some constant $c$ independent of the varying parameter.
\section{System Model}\label{systemmodel}
\subsection{Channel Model}
We consider a scenario where a BS with $N$ transmit and $M$ receive antennas serves a single-antenna user in the downlink, while simultaneously sensing a distant point target. The BS operates in a mono-static radar setup, i.e., it acts as both transmitter and receiver with identical angle of departure (AoD) and AoA, which holds under the assumption of far-field operation. The BS transmits $\mathbf{X}\in\mathbb{C}^{N\times L}$, a narrow-band dual-functional radar-communication signal, where $L>N$ is the frame length. The total transmit power is denoted by $p_t$. The user channel is modeled as Rician fading, consisting of a LoS component and a random NLoS component, i.e., $\tilde{\mathbf{h}}=\alpha_1\mathbf{a}(\theta_u)+\alpha_2\mathbf{h},$ with $\alpha_1=\sqrt{\frac{K}{1+K}}$, $\alpha_2=\sqrt{\frac{1}{1+K}}$, where $K>0$ is the Rician factor. The LoS component, $\mathbf{a}(\theta_u)\in\mathbb{C}^{N\times1}$, is the transmit steering vector toward the user and is given by
$\mathbf{a}(\theta_u)=\left[e^{-j\pi\sin(\theta_u)\frac{N-1}{2}}, \ldots, e^{j\pi\sin(\theta_u)\frac{N-1}{2}}\right]^T,$
where $\theta_u$ denotes the AoD toward the user, which may be either deterministic or random.
The NLoS component, $\mathbf{h}=[h_1,h_2,\ldots,h_N]^T\in\mathbb{C}^{N\times1}$, consists of i.i.d. entries $h_i=|h_i|e^{j\phi_i}$, where $|h_i|$ follows a Rayleigh distribution with scale parameter $1$, and $\phi_i$ is uniformly distributed over $[-\pi,\pi)$.
For the target, the transmit and receive steering vectors are given by
$\mathbf{a}(\theta)=\left[e^{-j\pi\sin(\theta)\frac{N-1}{2}}, \ldots, e^{j\pi\sin(\theta)\frac{N-1}{2}}\right]^T,$
and
$\mathbf{b}(\theta)=\left[e^{-j\pi\sin(\theta)\frac{M-1}{2}}, \ldots, e^{j\pi\sin(\theta)\frac{M-1}{2}}\right]^T,$
where $\theta$ denotes the AoD toward the target. The angles $\theta$ and $\theta_u$ are assumed to follow an arbitrary joint distribution, denoted by $f_{\theta,\theta_u}(\theta,\theta_u)$, thereby allowing for arbitrary correlation between the target and user directions. The received signal at the user is
\begin{align}
\mathbf{y}_u=c\tilde{\mathbf{h}}^H\mathbf{X}+\mathbf{z}_u,\label{recievdusersignal}
\end{align}
where $\mathbf{z}_u\in\mathbb{C}^{1\times L}$ is additive white Gaussian noise (AWGN) with independent elements distributed as $\mathcal{CN}(0,\sigma_u^2)$, and $c$ denotes the path gain. The radar echo signal at the BS is
\begin{align}
\mathbf{Y}_r=\alpha\mathbf{b}(\theta)\mathbf{a}(\theta)^H\mathbf{X}+\mathbf{Z}_r,
\end{align}
where $\alpha\in\mathbb{C}$ is the complex channel coefficient, which depends on the target radar cross section (RCS) and path loss, and $\mathbf{Z}_r\in\mathbb{C}^{M\times L}$ is AWGN with independent elements distributed as $\mathcal{CN}(0,\sigma_r^2)$.
\subsection{Performance Metrics and Problem Formulation}\label{newadd}
The performance of sensing and communication is evaluated using the CRB for angle estimation and the achievable rate $R$, respectively. Both $\mathrm{CRB}(\theta)$ and $R$ are random variables due to channel randomness and the channel-dependent design of $\mathbf{X}$. 
Following \cite{MIMOIntegratedSensingandCommunicationCRBRateTradeoff,CrameRaoBoundOptimizationforJoint,PhysicalLayerSecurityOptimizationWithCramér–RaoBoundMetric,OnStochasticFundamentalLimitsina}, $\mathbf{X}$ is designed based on perfect instantaneous CSI rather than average channel statistics, and its performance is evaluated by averaging over channel realizations to characterize long-term network behavior. For the sensing metric, we adopt the \textit{sensing outage probability (OP)}, $P_{\mathrm{out},s}(\epsilon)=\Pr\big(\mathrm{CRB}(\theta) > \epsilon\big)$ \cite{OnStochasticFundamentalLimitsina,OnStochasticPerformanceAnalysisofSecureIntegratedSensingandCommunicationNetworks}, which is suitable for real-time applications requiring guaranteed estimation accuracy. For the communication metric, we use the well-known communication OP, $P_{\mathrm{out},c}(\gamma)=\Pr\big(\mathrm{SINR} < \gamma\big)$.

Our objective is to determine the Pareto boundary of the achievable OP region for the ISAC system described above, characterizing the fundamental sensing–communication trade-off where neither performance metric can be improved without degrading the other. Following \cite{MIMOIntegratedSensingandCommunicationCRBRateTradeoff,onthefundementaltradeoff}, this boundary is obtained by minimizing the sensing CRB while varying the communication rate threshold and then translating the results into the OP domain.

We employ two precoding strategies: 1) Subspace joint beamforming (SJB), which adopts the structural form established in \cite[Lemma 1]{CrameRaoBoundOptimizationforJoint}. There, it was shown that for any fixed channel realization, the CRB-minimizing beamformer under SINR constraints lies in the span of the target steering vector and the user channel. In our work, we apply this structure to each realization of the random Rician fading channel, leading to random performance metrics that we analyze via OP. 2) Linear beamforming (LB), which uses separate beamforming vectors for communication and sensing. Although it does not strictly satisfy the optimality condition in \cite[Lemma 1]{CrameRaoBoundOptimizationforJoint}, it fully exploits the MIMO radar degrees of freedom and can potentially outperform SJB in certain scenarios \cite{OnStochasticFundamentalLimitsina}.
\subsection{Subspace Joint Beamforming (SJB)}\label{jointbeamforming}
In the SJB approach, the signal transmitted by the BS is defined as
\begin{align}
\mathbf{X}=\sqrt{p_t}\mathbf{w}\mathbf{s},\label{x}
\end{align}
where $\mathbf{w}\in\mathbb{C}^{N\times1}$ is the beamforming vector, and $\mathbf{s}\in\mathbb{C}^{1\times L}$ is the white Gaussian symbol vector for the user with i.i.d. entries. The average energy per symbol is $\mathbb{E}[|s_i|^2]=1$ for all $i$, and for sufficiently large $L$, we have $\frac{1}{L}\mathbf{s}\mathbf{s}^H\approx1$ \cite{CrameRaoBoundOptimizationforJoint}. The SJB beamformer $\mathbf{w}$ is designed to lie in the span of the user channel vector and the target steering vector:
\begin{align}
\mathbf{w}=\frac{b_1\tilde{\mathbf{h}}+b_2\mathbf{a}(\theta)}{||b_1\tilde{\mathbf{h}}+b_2\mathbf{a}(\theta)||},\label{wi}
\end{align}
where $b_1=|b_1|e^{j\phi_1}\in\mathbb{C}$ and $b_2=|b_2|e^{j\phi_2}\in\mathbb{C}$ are design parameters that can be optimized to minimize the CRB subject to the user's SINR being above a given threshold.
\subsection{Linear Beamforming (LB)}\label{linearbf}
In LB, the BS uses separate beamformers for the user and the radar probing signal, i.e., $\mathbf{w}_1=\frac{\tilde{\mathbf{h}}}{||\tilde{\mathbf{h}}||}\in\mathbb{C}^{N\times1},
\mathbf{w}_2=\frac{\mathbf{a}(\theta)}{||\mathbf{a}(\theta)||}\in\mathbb{C}^{N\times1}.$
The transmitted signal is given by
\begin{align}
\mathbf{X}=c_1\mathbf{w}_1\mathbf{s}_u+c_2\mathbf{w}_2\mathbf{s}_r,\label{xnew}
\end{align}
where $\mathbf{s}_u\in\mathbb{C}^{1\times L}$ is the white Gaussian symbol vector with i.i.d. entries for the user, with unit power, satisfying $\frac{1}{L}\mathbb{E}{\mathbf{s}_u\mathbf{s}_u^H}=1$ for large $L$ \cite{CrameRaoBoundOptimizationforJoint}, and $\mathbf{s}_r\in\mathbb{C}^{1\times L}$ is the dedicated radar signal with $\frac{1}{L}\mathbf{s}_r\mathbf{s}_r^H=1$. The parameters $c_1,c_2\in\mathbb{C}$ are design parameters. Assuming $\mathbf{s}_r$ and $\mathbf{s}_u$ are independent zero-mean sequences, the sample cross-term $\frac{1}{L}\mathbf{s}_u\mathbf{s}_r^H$ tends to zero as $L$ becomes sufficiently large by the law of large numbers. Therefore, for large $L$, the sample covariance matrix is
\begin{align}
\mathbf{R}_x=\frac{1}{L}\mathbf{X}\mathbf{X}^H\approx |c_1|^2\mathbf{w}_1\mathbf{w}_1^H+|c_2|^2\mathbf{w}_2\mathbf{w}_2^H.\label{rxnew}
\end{align}
Since the total transmit power is $p_t$, we have $p_t=\operatorname{Tr}(\mathbf{R}_x)=|c_1|^2+|c_2|^2.$
\begin{figure*}[!b]
\hrulefill
\normalsize
\begin{align}
%x_i&=|b_1|\cos(\phi_1)(\alpha^2_1+\alpha^2_2|h_i|^2+2\alpha_1\alpha_2|h_i|\cos(\phi_i+f^u_i))+|b_2|\alpha_1\cos(\phi_2-f_i+f^u_i)+|b_2|\alpha_2|h_i|\cos(\phi_2-f_i-\phi_i),\nonumber\\
%y_i&=|b_1|\sin(\phi_1)(\alpha^2_1+\alpha^2_2|h_i|^2+2\alpha_1\alpha_2|h_i|\cos(\phi_i+f^u_i))+|b_2|\alpha_1\sin(\phi_2-f_i+f^u_i)+|b_2|\alpha_2|h_i|\sin(\phi_2-f_i-\phi_i),\nonumber\\
%k_i&=\alpha^2_1|b_1|^2+\alpha^2_2|b_1|^2|h_i|^2+|b_2|^2+2\alpha_1\alpha_2|b_1|^2|h_i|\cos(\phi_i+f^u_i)+|b_1b_2|\alpha_1\cos(\phi_1-\phi_2+f_i-f^u_i)\nonumber\\
%&+2\alpha_2|b_1b_2|h_i|\cos(\phi_1-\phi_2+f_i+\phi_i).\label{xyk}\\
&\mathrm{var}(x_i)=|b_1|^2\kappa^2(\alpha^4_2+2\alpha^2_1\alpha^2_2)+\frac{\alpha^2_2|b_2|^2}{2}+2|b_1b_2|\alpha_1\alpha^2_2\kappa\cos(f^u_i+\phi_2-f_i),\nonumber\\
&\mathrm{cov}(x_iy_i)=|b_1|^2\kappa\zeta (\alpha^4_2+2\alpha^2_1\alpha^2_2)+|b_1b_2|\alpha^2_2\alpha_1\sin(\phi_1+\phi_2+f^u_i-f_i), 
\nonumber\\
&\mathrm{cov}(x_ik_i)=\alpha^2_2|b_1|\kappa(|b_1|^2(1+\alpha^2_1)+|b_2|^2)+2\kappa|b_1|^2|b_2|\alpha_1\alpha^2_2\cos(\phi_1-\phi_2+f_i-f^u_i)+\alpha_1\alpha^2_2|b_1|^2|b_2|\cos(f^u_i+\phi_2-f_i), 
\nonumber\\
&\mathrm{var}(y_i)=|b_1|^2\zeta^2(\alpha^4_2+2\alpha^2_1\alpha^2_2)+\frac{\alpha^2_2|b_2|^2}{2}+2|b_1b_2|\alpha_1\alpha^2_2\zeta\sin(f^u_i+\phi_2-f_i), 
\nonumber\\
&\mathrm{cov}(y_ik_i)=\alpha^2_2|b_1|\zeta(|b_1|^2(1+\alpha^2_1)+|b_2|^2)+2\zeta|b_1|^2|b_2|\alpha_1\alpha^2_2\cos(\phi_1-\phi_2+f_i-f^u_i)+\alpha_1\alpha^2_2|b_1|^2|b_2|\sin(f^u_i+\phi_2-f_i), 
\nonumber\\
&\mathrm{var}(k_i)=|b_1|^4(\alpha^4_2+2\alpha^2_1\alpha^2_2)+2\alpha^2_2|b_2|^2|b_1|^2+4|b_1b_2||b_1|^2\alpha_1\alpha^2_2\cos(f^u_i+\phi_2-f_i-\phi_1)\label{elemntofsigma}
\end{align}
\end{figure*}
\section{SJB performance analysis}\label{jointbeamformingperformance}
In this section, we analyze the performance of the SJB approach presented in Section~\ref{systemmodel} by evaluating its sensing and communication metrics.
\subsection{Communication OP}\label{opofuser}
Based on (\ref{recievdusersignal}), (\ref{x}), and (\ref{wi}), the SINR at the user is
\begin{align}
&\text{SINR}=\frac{p_tc^2}{\sigma^2_u} \frac{|\tilde{\mathbf{h}}^H(b_1\mathbf{\tilde{h}}+b_2\mathbf{a}(\theta))|^2}{|| b_1\mathbf{\tilde{h}}+b_2\mathbf{a}(\theta) ||^2}\nonumber\\
&\overset{(a)}{=}\frac{p_tc^2}{\sigma^2_u}\frac{(\sum_{i=1}^{N}x_i)^2+(\sum_{i=1}^{N}y_i)^2}{(\sum_{i=1}^{N}k_i)}\overset{(b)}{=}\frac{p_tc^2}{\sigma^2_u}\frac{X^2+Y^2}{Z},\label{sinr2i}
\end{align}
where (a) follows from defining the random variables \(x_i \triangleq \mathcal{R}\left(b_1|\tilde{h}_i|^2 + b_2\tilde{h}^*_i e^{-jf_i}\right)\), \(y_i \triangleq \mathcal{I}\left(b_1|\tilde{h}_i|^2 + b_2\tilde{h}^*_i e^{-jf_i}\right)\), \(k_i \triangleq |b_1\tilde{h}_i + b_2 e^{-jf_i}|^2\). 
\begin{comment}
These RVs are expanded in (\ref{xyk}) at the bottom of this page.
\end{comment} 
Here, \(\tilde{h}_i\) is the \(i\)-th element of the vector \(\mathbf{\tilde{h}}\). Moreover, \(f^u_i = \pi \sin(\theta_u) \frac{N - (2i - 1)}{2}\) and \(f_i = \pi \sin(\theta) \frac{N - (2i - 1)}{2}\). We note that $x_i$, $y_i$, and $k_i$ depend on the same underlying random variables: $f_i$ (a function of $\theta$), $f_i^u$ (a function of $\theta_u$), $|h_i|$, and $\phi_i$. ($b$) is obtained by defining $X=\sum_{i=1}^{N}x_i$, $Y=\sum_{i=1}^{N}y_i$ and $Z=\sum_{i=1}^{N}k_i$. Therefore, the communication OP is given by
\begin{align}
P_{\text{out}, c}(\gamma)\!{=}\!\!\!\iint_{0}^{\pi}\!\!\!\!\!\!P(\frac{p_t|c|^2}{\sigma^2_u}\frac{X^2+Y^2}{Z}<\gamma)|\theta,\theta_u)f_{\theta,\theta_u}(\theta,\theta_u)\mathrm{d}\theta \mathrm{d}\theta_u. \label{newversion}
\end{align}
To calculate the inner probability, we need to derive the joint PDF of $X$, $Y$, and $Z$ conditioned on $\theta$ and $\theta_u$. We note that $X$, $Y$, and $Z$ (also $x_i$, $y_i$, and $k_i$) are not independent, as they are functions of $h_i$, $f_i$, and $f^u_i$. In the following analysis, in accordance with (\ref{newversion}), we consider the angles $\theta$ and $\theta_u$ as constants. We define $N$ random vectors, $\mathbf{d}_i=[x_i, y_i, k_i]^T \in \mathbb{R}^{3 \times 1}$, for all $i=1,...,N$. For any pair of $j$ and $i\neq j$, the vectors $\mathbf{d}_j$ and $\mathbf{d}_i$ are independent because the RVs $h_i$s are i.i.d. The mean vector and covariance matrix of $\mathbf{d}_i$ are provided in the following Lemma.
\begin{lemma}\label{lemma1i}
The mean vector and covariance matrix of $\mathbf{d}_i$ are given by $\boldsymbol{\mu}_i=[|b_1| \cos(\phi_1)+|b_2| \alpha_1 \cos(\phi_2-f_i+f^u_i),|b_1| \sin(\phi_1)+|b_2| \alpha_1 \sin(\phi_2-f_i+f^u_i),|b_1|^2 + |b_2|^2+2 \alpha_1 |b_1b_2|\cos(\phi_1-f^u_i-\phi_2+f_i)]$ and 
$\mathbf{\Sigma}_i=\begin{bmatrix}
\mathrm{var}(x_i) & \mathrm{cov}(x_iy_i) & \mathrm{cov}(x_ik_i)\\
\mathrm{cov}(x_iy_i) & \mathrm{var}(y_i) & \mathrm{cov}(y_ik_i)\\
\mathrm{cov}(x_ik_i) & \mathrm{cov}(y_ik_i) &\mathrm{var}(k_i)
\end{bmatrix}$, respectively. By defining $\zeta \triangleq \sin(\phi_1)$ and $\kappa \triangleq \cos(\phi_1)$, the elements of $\mathbf{\Sigma}_i$ are given in (\ref{elemntofsigma}) at the bottom of this page.
\end{lemma}
\begin{proof}
The proof follows a similar approach as \cite[Appendix A]{OnStochasticFundamentalLimitsina}.
\end{proof}
We note that, as shown by Lemma~\ref{lemma1i}, the vectors $\mathbf{d}_i$ are not identically distributed, since their mean and covariance depend on the index $i$ through $f_i$ or $f_i^u$. When $N$ is large, which is typically the case in MIMO-ISAC systems due to the use of large antenna arrays, the independence of the random vectors $\mathbf{d}_i$ across the array allows the application of the CLT, thereby justifying the Gaussian approximation
\begin{align}
\sum_{i=1}^N \mathbf{d}_i \stackrel{\text{d}}{\approx} \mathcal{N}\left( \sum_{i=1}^N \boldsymbol{\mu}_i, \sum_{i=1}^N \boldsymbol{\Sigma}_i \right). \label{multyclt}
\end{align}
whose probability density function is denoted by $f(x, y, z)$. Moreover, by defining 
$C_1 \triangleq \sum_{i=1}^{N} \cos(f^u_i + \phi_2 - f_i)$, $S_1 \triangleq \sum_{i=1}^{N} \sin(f^u_i + \phi_2 - f_i)$, $C_2 \triangleq \sum_{i=1}^{N} \cos(\phi_1 - \phi_2 + f_i - f^u_i)$, and $S_2 \triangleq \sum_{i=1}^{N} \sin(\phi_1 + \phi_2 + f_i - f^u_i)$, the elements of $\sum_{i=1}^N \boldsymbol{\mu}_i$ and $\sum_{i=1}^N \boldsymbol{\Sigma}_i$ can be presented into closed-form, computationally efficient expressions using the following lemma (the proof is provided in Appendix \ref{lemmanewricianproof}).
\begin{lemma}\label{lemmanewrician}
We have $C_1= \frac{\sin(N\Delta)}{\sin(\Delta)} \cos(\phi_2)$, $S_1=\frac{\sin(N\Delta)}{\sin(\Delta)} \sin(\phi_2)$, $C_2=\frac{\sin(N\Delta)}{\sin(\Delta)} \cos(\phi_1 - \phi_2)$, and $S_2=\frac{\sin(N\Delta)}{\sin(\Delta)} \sin(\phi_1 + \phi_2)$, where $\Delta = \frac{\pi}{2}[\sin(\theta_u) - \sin(\theta)]$.
\end{lemma}
\textit{Note:} The expressions \(\frac{\sin(N\Delta)}{\sin\Delta}\) are understood in the limiting sense when \(\sin\Delta = 0\) (i.e., when \(\Delta = m\pi\)). At these points, the sums evaluate to \(N\cos(\phi_2)\) for \(C_1\), \(N\sin(\phi_2)\) for \(S_1\), etc., which can be verified by direct substitution or by taking limits. These cases correspond to \(\sin(\theta_u) = \sin(\theta)\) or \(\{\theta_u, \theta\} = \{\pm\frac{\pi}{2}, \mp\frac{\pi}{2}\}\), which are measure-zero events and do not affect the general analysis.

Thus, using Lemma \ref{lemma1i} and Lemma \ref{lemmanewrician} we have $\sum_{i=1}^N \boldsymbol{\mu}_i = \begin{bmatrix}
N|b_1|\cos(\phi_1) + |b_2|\alpha_1 C_1 \\
N|b_1|\sin(\phi_1) + |b_2|\alpha_1 S_1 \\
N(|b_1|^2 + |b_2|^2) + 2\alpha_1|b_1b_2| C_2
\end{bmatrix}$
and the elements of $\sum_{i=1}^N \boldsymbol{\Sigma}_i$ can be obtained in closed form. For instance, $\sum \text{var}(x_i) = N\left[|b_1|^2\kappa^2(\alpha_2^4 + 2\alpha_1^2\alpha_2^2)+ \frac{\alpha_2^2|b_2|^2}{2}\right] + 2|b_1b_2|\alpha_1\alpha_2^2\kappa \cdot C_1$. The remaining elements follow similarly and can be expressed in terms of $C_1, S_1, S_2,$ and $C_2$. They are omitted here due to space limitations.

Thus, by defining $I(f(x,y,z),\mathcal{D}) \triangleq 
\int_{0}^{\pi} \int_{0}^{\pi} 
[ 
\iiint\limits_{\mathcal{D}} f(x, y, z| \theta,\theta_u) \, \mathrm{d}x\, \mathrm{d}y\, \mathrm{d}z
] 
f_{\theta,\theta_u}(\theta,\theta_u)\, \mathrm{d}\theta\, \mathrm{d}\theta_u,$ we have $P_{\text{out}, c}(\gamma) = I(f(x,y,z), \mathcal{D})$, where $\mathcal{D} \triangleq \left\{(X,Y,Z) : \frac{p_t|c|^2(X^2 + Y^2)}{Z \sigma_u^2} < \gamma\right\}$.  In general, integrating a multivariate Gaussian PDF over an arbitrary region lacks a general analytical expression. Therefore, numerical methods are necessary for its evaluation \cite{Amethodtointegrate}.
\begin{remark}
The Gaussian approximation in \eqref{multyclt} can be justified by the multivariate Lyapunov CLT via the Cram\'er--Wold device. For any fixed $\mathbf t\in\mathbb R^3$, define
$
Z_{i,\mathbf t}\triangleq \mathbf t^T(\mathbf d_i-\boldsymbol\mu_i),\qquad
s_N^2(\mathbf t)\triangleq \sum_{i=1}^N \mathrm{Var}(Z_{i,\mathbf t}).
$
Since $\tilde h_i=\alpha_1 e^{-jf_i^u}+\alpha_2 h_i$, each component of $\mathbf d_i$ is a real-valued polynomial of total degree at most two in $h_i$ and $\bar h_i$, with deterministic coefficients depending on $f_i,f_i^u,b_1,b_2$. These coefficients are uniformly bounded in $i$, and $h_i$ has finite moments of all orders. Hence, for every fixed $\mathbf t$, there exists a constant $C_{\mathbf t}<\infty$ such that
$
\sup_i \mathbb E\!\left[|Z_{i,\mathbf t}|^3\right]\le C_{\mathbf t},
$
which implies
$
\sum_{i=1}^N \mathbb E\!\left[|Z_{i,\mathbf t}|^3\right]=O(N).
$
Moreover, under nondegenerate parameter choices, for every fixed nonzero $\mathbf t$ we have
$
s_N^2(\mathbf t)=\sum_{i=1}^N \mathrm{Var}(Z_{i,\mathbf t})=\Theta(N),
$
since the summands are independent, each $\mathrm{Var}(Z_{i,\mathbf t})$ is uniformly bounded, and their average stays bounded away from zero. Therefore,
$
\frac{\sum_{i=1}^N \mathbb E\!\left[|Z_{i,\mathbf t}|^3\right]}
{\big(s_N^2(\mathbf t)\big)^{3/2}}
=O(N^{-1/2})\to 0.
$
Thus, by the scalar Lyapunov CLT applied to $Z_{i,\mathbf t}$ and the Cram\'er--Wold device, $\sum_{i=1}^N \mathbf d_i$ is asymptotically Gaussian with mean $\sum_{i=1}^N\boldsymbol\mu_i$ and covariance $\sum_{i=1}^N\boldsymbol\Sigma_i$.
\end{remark}
\textbf{Special Cases of Interest:}\label{simplified}
In the following, we consider five special cases which result in a simplified final expression for $P_{\text{out}, c}(\gamma)$.

1) The case $K=0$ corresponds to Rayleigh fading channels to the user. In this case, the result reduces to that of \mbox{\cite[Lemma 3]{OnStochasticFundamentalLimitsina}}.

2) The case $K \rightarrow \infty$ corresponds to pure LoS channels to the user. In this case, after some algebraic calculation we have $P_{\mathrm{out},c}(\gamma)=\iint_{\mathcal D(\gamma)}
f_{\theta,\theta_u}(\theta,\theta_u)
\mathrm{d}\theta_u
\mathrm{d}\theta$, where $\mathcal D(\gamma)
= \{ (\theta,\theta_u) \in [0,\pi]^2 :
\frac{p_t|c|^2}{\sigma_u^2}
\cdot
\frac{\left| b_1 N + b_2 \frac{\sin(N\Delta)}{\sin(\Delta)} \right|^2}
{N(|b_1|^2 + |b_2|^2)
+ 2|b_1 b_2| \frac{\sin(N\Delta)}{\sin(\Delta)}\cos(\phi_2 - \phi_1)}< \gamma\}$.  If either the target angle or the user angle is fixed, the integral becomes a 1-D integral over the domain of the remaining angle. 

3) If $\theta_u = \theta \ (f_i = f^u_i)$, i.e., the target and user share the same LoS direction (while the user still experiences an NLoS channel $\mathbf{h}$) then by Lemma \ref{lemma1i}, $\mathbf{d}_i$s are i.i.d., and $\sum_{i=1}^N \mathbf{d}_i \stackrel{\text{d}}{\approx} \mathcal{N}\left( N \boldsymbol{\mu}, N \boldsymbol{\Sigma} \right)$ whose probability density function is denoted by $\tilde{f}(x, y, k)$, where the elements of $\boldsymbol{\mu}$ and $\boldsymbol{\Sigma}$ follow Lemma \ref{lemma1i} with $f_i=f^u_i$, making them independent of $\theta$ and $\theta_u$ (but dependent on the Rician factors $\alpha_1$ and $\alpha_2$). Thus,
$P_{\text{out}, c}(\gamma)=\iiint\limits_{\mathcal D_1(\gamma)} \tilde{f}(x, y, z) \mathrm{d}x \, \mathrm{d}y \, \mathrm{d}z$, where $\mathcal D_1(\gamma)
\triangleq \left\{(X,Y,Z) : \frac{p_t|c|^2}{\sigma^2_u}\frac{X^2+Y^2}{Z}<\gamma \right\}$.  Following a similar approach as in \cite[Appendix B]{OnStochasticFundamentalLimitsina}, the communication OP can be expressed in closed form as \(P_{\mathrm{out},c}(\gamma) = F_{w,k,\lambda,s,m}(0)\), where \(F_{w,k,\lambda,s,m}(\cdot)\) is the CDF of the generalized chi-square distribution. The parameters \(w, k, \lambda, s, m\) are computed from the mean vector \(\boldsymbol{\mu}\) and covariance matrix \(\boldsymbol{\Sigma}\) given in Lemma \ref{lemma1i} (with \(f_i = f_i^u\)), following the transformation steps detailed in \cite[Appendix B]{OnStochasticFundamentalLimitsina}. 

4) Another simplification arises when the relative user-target angular separation falls on a null of the array coupling term, namely when
$
\sin(\theta_u)-\sin(\theta)=\frac{2n}{N},\qquad n=\pm1,\pm2,\ldots,
$
so that
$
\frac{\sin(N\Delta)}{\sin(\Delta)}=0,
\qquad
\Delta=\frac{\pi}{2}\big[\sin(\theta_u)-\sin(\theta)\big].
$
Intuitively, this condition means that the coherent LoS coupling between the user and target steering directions cancels across the array. As a result, the angle-dependent sums appearing in $\sum_{i=1}^N \boldsymbol{\mu}_i$ and $\sum_{i=1}^N \boldsymbol{\Sigma}_i$ vanish, i.e., $C_1=S_1=C_2=S_2=0$. Therefore, although the vectors $\mathbf d_i$ are still not i.i.d., the total mean $\sum_{i=1}^N\boldsymbol{\mu}_i$ and covariance $\sum_{i=1}^N\boldsymbol{\Sigma}_i$ reduce to angle-independent closed-form expressions (but they remain dependent on the Rician factors $\alpha_1$ and $\alpha_2$). Thus, the OP is derived by following a similar approach as in Case 3.

5) For highly concentrated angle distributions (low angular spread scenarios), where both user and target angles exhibit very small variance around their mean values, simplification of the OP calculation is provided by using the following lemma.
\begin{figure*}[!b]
\hrulefill
\normalsize
\begin{align}
&\text{CRB}(\theta,\alpha,\sigma_r)=\frac{\sigma^2_r \text{Tr}(\mathbf{A}^H(\theta) \mathbf{A}(\theta) \mathbf{R}_x)}{2 |\alpha|^2 L (\text{Tr}(\mathbf{A}^H(\theta) \mathbf{A}(\theta) \mathbf{R}_x) \text{Tr}(\dot{\mathbf{A}}^{H}(\theta) \dot{\mathbf{A}}^(\theta) \mathbf{R}_x)-|\text{Tr}(\dot{\mathbf{A}}^{H}(\theta) \mathbf{A}(\theta) \mathbf{R}_x)|^2)}. \label{crb}
\end{align}
%\begin{align}
%\tilde{x}_i&=|b_1|\alpha_1\cos(\phi_1+f_i-f^u_i)+|b_1|\alpha_2|h_i|\cos(\phi_1+f_i+\phi_i)+|b_2|\cos(\phi_2),\nonumber\\
%\tilde{y}_i&=|b_1|\alpha_1\sin(\phi_1+f_i-f^u_i)+|b_1|\alpha_2|h_i|\sin(\phi_1+f_i+\phi_i)+|b_2|\sin(\phi_2),\label{xyktilde}\\
%\end{align}
\end{figure*}
\begin{lemma}\label{cheb}
If \(\theta = \theta_0 + \epsilon_1\), \(\theta_u = \theta_{u0} + \epsilon_2\) with \(\mathbb{E}[\epsilon_i] = 0\), \text{Var}\((\epsilon_i) = \sigma^2\), and \(\sin(\Delta_0) \neq 0\) where \(\Delta_0 = \frac{\pi}{2}[\sin(\theta_{u0}) - \sin(\theta_0)]\), then
$\lim_{\sigma \to 0} \frac{\sin(N\Delta)}{\sin(\Delta)} = \frac{\sin(N\Delta_0)}{\sin(\Delta_0)} \quad \text{in probability}$.
\end{lemma}
\begin{proof}
Using the Taylor expansion around the mean angles and defining $\Delta_0 = \frac{\pi}{2}[\sin(\theta_{u0}) - \sin(\theta_0)]$ and $\delta = \frac{\pi}{2}[\epsilon_2 \cos(\theta_{u0}) - \epsilon_1 \cos(\theta_0)]$, we have $\Delta = \Delta_0 + \delta + O(\sigma^2)$. Then $\frac{\sin(N\Delta)}{\sin(\Delta)}= \frac{\sin(N\Delta_0)}{\sin(\Delta_0)} + \delta \cdot \frac{d}{d\Delta}\left(\frac{\sin(N\Delta)}{\sin(\Delta)}\right)\Big|_{\Delta=\Delta_0} + O(\delta^2)$. Also, we have $\frac{d}{d\Delta}\left(\frac{\sin(N\Delta)}{\sin(\Delta)}\right) = \frac{N\cos(N\Delta)\sin(\Delta) - \sin(N\Delta)\cos(\Delta)}{\sin^2(\Delta)}$. Since \(\delta = O(\sigma)\) and \(\mathbb{E}[\delta] = 0\), we have $\mathbb{E}[\frac{\sin(N\Delta)}{\sin(\Delta)}] \approx \frac{\sin(N\Delta_0)}{\sin(\Delta_0)} + O(\sigma^2)$ and $\text{Var}(\frac{\sin(N\Delta)}{\sin(\Delta)}) \approx \text{Var}(\delta)  \left|\frac{N\cos(N\Delta_0)\sin(\Delta_0) - \sin(N\Delta_0)\cos(\Delta_0)}{\sin^2(\Delta_0)}\right|^2 + O(\sigma^3)=O(\sigma^2)$. Therefore, as $\sigma \rightarrow 0$, by Chebyshev's inequality we have $P(|\frac{\sin(N\Delta)}{\sin(\Delta)} - \frac{\sin(N\Delta_0)}{\sin(\Delta_0)}| \geq \epsilon) \leq \frac{\text{Var}(\frac{\sin(N\Delta)}{\sin(\Delta)})}{\epsilon^2} \to 0$. This implies that \(\frac{\sin(N\Delta)}{\sin(\Delta)} \overset{p}{\to} \frac{\sin(N\Delta_0)}{\sin(\Delta_0)}\).
\end{proof}
Using Lemma \ref{cheb}, the angle-dependent sums appearing in the $\sum_{i=1}^N\boldsymbol{\mu}_i$ and $\sum_{i=1}^N\boldsymbol{\Sigma}_i$ can be approximated by deterministic constants, e.g.,
$
C_1 \approx \frac{\sin(N\Delta_0)}{\sin(\Delta_0)}\cos(\phi_2), \quad
S_1 \approx \frac{\sin(N\Delta_0)}{\sin(\Delta_0)}\sin(\phi_2),
$
$
C_2 \approx \frac{\sin(N\Delta_0)}{\sin(\Delta_0)}\cos(\phi_1-\phi_2), \quad
S_2 \approx \frac{\sin(N\Delta_0)}{\sin(\Delta_0)}\sin(\phi_1+\phi_2).
$
Thus, the OP is derived by following a similar
approach as in Case 3.
\subsection{Sensing OP}\label{opoftarget}
We utilize the CRB, an expression for which was derived in \cite[Appendix C]{TargetDetectionandLocalization}, for a given $\theta$, which is given by (\ref{crb}) and appears at the bottom of this page, where $\mathbf{A}(\theta)= \mathbf{b}(\theta) \mathbf{a}^H(\theta)$, and $\dot{\mathbf{A}}$ is the derivative of ${\mathbf{A}}$ with respect to $\theta$. Instead of \(\text{CRB}(\theta, \alpha, \sigma_r)\), for brevity we will use the notation \(\text{CRB}(\theta)\) for the remainder of the paper. Inserting (\ref{x}) and (\ref{wi}) into (\ref{crb}) and after some algebraic derivations, the \text{CRB}($\theta$) can be expressed as
\begin{align}
&\text{CRB}(\theta)\!\!=\frac{\sigma_r^2}{2 L p_t|\alpha|^2 || \mathbf{b'} ||^2 | \mathbf{a}^H \mathbf{w} | ^2} \!\!\overset{(a)}{=}\frac{g(\theta)(\sum_{i=1}^{N}k_i)}{(\sum_{i=1}^{N}\tilde{x}_i)^2+(\sum_{i=1}^{N}\tilde{y}_i)^2}\nonumber\\
&\overset{(b)}{=}\frac{g(\theta)Z}{\tilde{X}^2+\tilde{Y}^2},\label{cramer}
\end{align}
where $\mathbf{b'}$ denotes the derivative of $\mathbf{b}$ with respect to $\theta$; $(a)$ and $(b)$ are due to defining \(g(\theta) \triangleq \frac{6 \sigma_r^2}{L p_t|\alpha|^2 (M-1)(M)(M+1)\pi^2 \cos^2(\theta)}\) and the RVs $\tilde{x}_i\triangleq \mathcal{R}(b_1e^{jf_i}\tilde{h}_i+b_2)$, $\tilde{y}_i\triangleq\mathcal{I}(b_1e^{jf_i}\tilde{h}_i+b_2)$,
$\tilde{X}=\sum_{i=1}^{N}\tilde{x}_i$, and $\tilde{Y}=\sum_{i=1}^{N}\tilde{y}_i$ (recall that \(k_i\) and $Z$ were defined in Subsection \ref{opofuser}).
\begin{comment}
$\tilde{x}_i$ and $\tilde{y}_i$ are expanded in (\ref{xyktilde}) at the bottom of this page.
\end{comment}
Therefore, to calculate $P_{\text{out}, s}(\epsilon)=P(\text{CRB}(\theta)>\epsilon)$, we need to derive the joint PDF of $\tilde{X}$, $\tilde{Y}$, and $Z$. Following a similar approach as in Subsection \ref{opofuser}, we define $N$ random vectors $\mathbf{\tilde{d}}_i = [\tilde{x}_i, \tilde{y}_i, k_i]^T \in \mathbb{R}^{3 \times 1}$, which are conditionally independent given $\theta$ and $\theta_u$. When $N$ is large, we have $[\tilde{X}, \tilde{Y}, Z]^T \stackrel{\text{d}}{\approx} \mathcal{N}\left( \sum_{i=1}^N \tilde{\boldsymbol{\mu}}_i,\ \sum_{i=1}^N \tilde{\boldsymbol{\Sigma}}_i \right),$
whose probability density function is denoted by $\hat{{f}}(\tilde{x}, \tilde{y}, z)$, where $\tilde{\boldsymbol{\mu}}_i$ and $\tilde{\boldsymbol{\Sigma}}_i$ are the mean vector and covariance matrix of $\tilde{\mathbf{d}}_i$, respectively, closed-form expressions for which are provided in the following Lemma.
\begin{lemma}\label{lemma2i}
The mean vector of $\tilde{\mathbf{d}}_i$ is given by $\tilde{\boldsymbol{\mu}}_i=[|b_1|\alpha_1\cos(\phi_1+f_i-f^u_i)+|b_2|\cos(\phi_2),\ |b_1|\alpha_1\sin(\phi_1+f_i-f^u_i)+|b_2|\sin(\phi_2),\ |b_1|^2 + |b_2|^2+2 \alpha_1 |b_1b_2|\cos(\phi_1-f^u_i-\phi_2+f_i)]$. Moreover, the covariance matrix of $\tilde{\mathbf{d}}_i$ is given by 
$\tilde{\boldsymbol{\Sigma}}_i=\begin{bmatrix}
\mathrm{var}(\tilde{x}_i) & \mathrm{cov}(\tilde{x}_i\tilde{y}_i) & \mathrm{cov}(\tilde{x}_ik_i)\\
\mathrm{cov}(\tilde{x}_i\tilde{y}_i) & \mathrm{var}(\tilde{y}_i) & \mathrm{cov}(\tilde{y}_i\tilde{k}_i)\\
\mathrm{cov}(\tilde{x}_i\tilde{k}_i) & \mathrm{cov}(\tilde{y}_ik_i) &\mathrm{var}(k_i)
\end{bmatrix}$, where its elements are given by $\mathrm{var}(\tilde{x}_i)=\frac{|b_1|^2\alpha^2_2}{2}$, $\mathrm{cov}(\tilde{x}_i\tilde{y}_i)=0$, $\mathrm{cov}(\tilde{x}_i{k}_i)=|b_1|^2|b_2|\alpha^2_2\cos(\phi_2)+\alpha_1\alpha^2_2|b_1|^3\cos(\phi_1+f_i-f^u_i)$, $\mathrm{var}(\tilde{y}_i)= \frac{|b_1|^2\alpha^2_2}{2}$, $\mathrm{cov}(\tilde{y}_ik_i)=|b_1|^2|b_2|\alpha^2_2\sin(\phi_2)+\alpha_1\alpha^2_2|b_1|^3\sin(\phi_1+f_i-f^u_i)$, $\mathrm{var}({k}_i)= |b_1|^4(\alpha^4_2+2\alpha^2_1\alpha^2_2)+2\alpha^2_2|b_2|^2|b_1|^2+4|b_1b_2||b_1|^2\alpha_1\alpha^2_2\cos(f^u_i+\phi_2-f_i-\phi_1)$.
\end{lemma}
\begin{proof}
The proof follows a similar approach as used in the proof of Lemma \ref{lemma1i}.
\end{proof}
Thus, we have $P_{\text{out}, s}(\epsilon) = I(\hat{f}(\tilde{x}, \tilde{y}, z),\mathcal{\tilde{D}})$, where $\mathcal{\tilde{D}} \triangleq \left\{(\tilde{X},\tilde{Y},Z) : \frac{g(\theta)Z}{\tilde{X}^2+\tilde{Y}^2} >\epsilon \right\}$. 

\textbf{Special Cases of Interest:}\label{simplified2}
In the following, we consider some special cases which result in a simplified final expression for $P_{\text{out}, s}(\epsilon)$.

1) The case $K=0$ corresponds to Rayleigh fading channels to the user. In this case, the result reduces to that of \cite[Lemma5]{OnStochasticFundamentalLimitsina}.

2) The case $K \rightarrow \infty$ corresponds to pure LoS channels to the user. In this case, after some algebraic calculation we have $P_{\mathrm{out},s}(\epsilon)=\iint_{\mathcal {\hat{D}}(\epsilon)}
f_{\theta,\theta_u}(\theta,\theta_u)
\mathrm{d}\theta_u
\mathrm{d}\theta$, where $\mathcal {\hat{D}}(\epsilon)
= \{ (\theta,\theta_u) \in [0,\pi]^2 :
g(\theta) \cdot \frac{N(|b_1|^2 + |b_2|^2) + 2|b_1b_2|\frac{\sin(N\Delta)}{\sin(\Delta)}\cos(\phi_2 - \phi_1)}{|b_1\frac{\sin(N\Delta)}{\sin(\Delta)} + b_2N|^2} > \epsilon\}$. If either the target angle or the user angle is fixed, the integral becomes a 1-D integral over the domain of the remaining angle.

3) If $\theta_u = \theta \ (f_i = f^u_i)$, or if $K \neq 0$ and $\sin(\theta_u) - \sin(\theta) = \frac{2n}{N} \text{ for } n = \pm1, \pm2, \ldots$, i.e., $\frac{\sin(N\Delta)}{\sin(\Delta)} \cos(\phi_1) = \frac{\sin(N\Delta)}{\sin(\Delta)} \sin(\phi_1) = \frac{\sin(N\Delta)}{\sin(\Delta)} \cos(\phi_1-\phi_2) = 0$, or for highly concentrated angle distributions, $P_{\text{out}, s}(\epsilon)$ is derived by following a similar approach as in Case 3, Case 4, and Case 5 of Section \ref{simplified}, respectively, using the corresponding PDF
derived from Lemma \ref{lemma2i} and setting the region of integration as $\mathcal{\tilde{D}}$.
\section{LB Performance Analysis}\label{linearbeamforing}
In this section, we analyze the performance of the LB approach presented in Section~\ref{systemmodel} by evaluating its sensing and communication metrics.
\subsection{Communication OP}\label{opuserlb}
In this subsection, we derive the communication OP with and without DPC. Based on (\ref{recievdusersignal}), (\ref{xnew}), and the Rician model for $\tilde{\mathbf{h}}$ described in Section \ref{systemmodel}, the received signal at the user can be expressed as follows:
\begin{equation}
\mathbf{y}_u= c\left(c_1|| (\alpha_1\mathbf{a}(\theta_u)+\alpha_2\mathbf{h}) ||\mathbf{s}_u+c_2
\frac{\tilde{\mathbf{h}}^H  \mathbf{a} \mathbf{s}_r}{|| \mathbf{a} ||}\right)+ \mathbf{z}_u.\label{yunew}
\end{equation}
Thus, the SINR of the user is
\begin{align}
&\text{SINR}\!\!= \frac{c^2|c_1|^2|| (\alpha_1\mathbf{a}(\theta_u)+\alpha_2\mathbf{h}) || ^2}{\sigma_u^2+\frac{|cc_2|^2}{||\mathbf{a}||^2}|\mathbf{a}^H(\alpha_1\mathbf{a}(\theta_u)+\alpha_2\mathbf{h})|^2}\!\!\nonumber\\
&\overset{(a)}{=}\!\!\frac{|c_1|^2U}{(\frac{\sigma_u^2}{c^2}+\frac{|c_2|^2}{N}(\hat{R}^2+\hat{T}^2))}\label{SINRnotdpc},
\end{align}
where $\hat{R} =\sum_{i=1}^N \hat{r}_i,$ and $\hat{T}= \sum_{i=1}^N \hat{t}_i,$ $U\triangleq \sum_{i=1}^{N}u_i$, where $\hat{r}_i = \Re\left( \alpha_1 e^{j (f_i - f^u_i)} + \alpha_2 e^{j f_i} h_i \right)= \alpha_1 \cos(f_i - f^u_i) + \alpha_2 |h_i| \cos(f_i + \phi_i),$ $\hat{t}_i = \Im\left( \alpha_1 e^{j (f_i -f^u_i)} + \alpha_2 e^{j f_i} h_i \right)=\alpha_1 \sin(f_i - f^u_i) + \alpha_2 |h_i| \sin(f_i + \phi_i)$, and $u_i\triangleq |\alpha_1e^{-jf^u_i}+\alpha_2h_i|^2=\alpha^2_1+\alpha^2_2|h_i|^2+2\alpha_1\alpha_2|h_i|\cos(f^u_i+\phi_i)$.

Moreover, following a similar approach as in Subsection \ref{opofuser}, we have ${\hat{\mathbf{d}}} = [\hat{R}, \hat{T}, U]^T \stackrel{\text{d}}{\approx} \mathcal{N}\left( \sum_{i=1}^N {\hat{\boldsymbol{\mu}}}_i,\ \sum_{i=1}^N {\hat{\boldsymbol{\Sigma}}}_i \right),$
whose probability density function is denoted by $\check{f}(\hat{R}, \hat{T}, U)$,
where ${\hat{\boldsymbol{\mu}}}_i$ and ${\hat{\boldsymbol{\Sigma}}}_i$ are the mean vector and covariance matrix of ${\hat{\mathbf{d}}}_i = [\hat{r}_i, \hat{t}_i, u_i]$, respectively. Following a similar approach as in the proof of Lemma~\ref{lemma1i}, we obtain the following lemma.
\begin{lemma}\label{lemma5i}
The mean vector of $\hat{\mathbf{d}}_i$ is given by
${\hat{\boldsymbol{\mu}}}_i=[\alpha_1\cos(f_i-f^u_i),\ \alpha_1\sin(f_i-f^u_i),\ 1]^T$.
Moreover, the covariance matrix ${\hat{\boldsymbol{\Sigma}}}_i$ has elements
$\mathrm{var}(\hat{r}_i)=\frac{\alpha^2_2}{2}$,
$\mathrm{var}(\hat{t}_i)=\frac{\alpha^2_2}{2}$,
$\mathrm{var}(u_i)=\alpha^4_2+2\alpha^2_1\alpha^2_2$,
$\mathrm{cov}(\hat{r}_i,\hat{t}_i)=0$,
$\mathrm{cov}(\hat{r}_i,u_i)=\alpha_1\alpha^2_2\cos(f_i-f^u_i)$,
$\mathrm{cov}(\hat{t}_i,u_i)=\alpha_1\alpha^2_2\sin(f_i-f^u_i)$.
\end{lemma}
Therefore, the communication OP is $P_{\text{out}, c}(\gamma) = I(\check{f}(\hat{R}, \hat{T},U),\mathcal{{\check{D}}})$, where 
$\mathcal{{\check{D}}}
\triangleq \left\{(\hat{R},\hat{T},U): \frac{|c_1c|^2U}{(\sigma_u^2+\frac{|cc_2|^2}{N}(\hat{R}^2+\hat{T}^2))}<\gamma)\right\}$.

Moreover, given that $\mathbf{s}_r$ is known to the BS, we can apply the DPC theorem \cite{elgamal}. Thus, following the proof in \cite[Appendix C]{OnStochasticFundamentalLimitsina}, the instantaneous rate of the user when receiving the signal given by (\ref{yunew}) is equal to $\log_{2}(1+ \frac{|c_1c|^2}{\sigma^2_u}|| (\alpha_1\mathbf{a}(\theta_u)+\alpha_2\mathbf{h}) ||^2)=\log_{2}(1+ U\frac{|c_1c|^2}{\sigma^2_u})$ when applying DPC. Thus, using Lemma \ref{lemma5i}, we have $U \stackrel{\text{d}}{\approx} \mathcal{N}\left( N, \ N(\alpha_2^4 + 2\alpha_1^2\alpha_2^2) \right)$ (note that this distribution is independent of $\theta$ and $\theta_u$). Thus,
the OP is
$P\left( U < \frac{\gamma\sigma^2_u}{|c_1c|^2} \right)=\Phi\left( \frac{\frac{\gamma\sigma^2_u}{|c_1c|^2} - N}{\sqrt{N(\alpha_2^4 + 2\alpha_1^2\alpha_2^2)}} \right)$, where \(\Phi\) is the CDF of the normalized Gaussian distribution.
\subsection{Sensing OP}\label{opoftargetlb}
By replacing (\ref{rxnew}) in (\ref{crb}), the CRB is derived in the following lemma.
\begin{lemma}\label{lemma4}
The CRB for the case of LB is
\begin{align}
\text{CRB}(\theta)\!\!=\!\!\frac{\frac{Q}{||\mathbf{b}'||^2}U \psi}{\psi^2+\!\frac{|c_1c_2|^2}{||\mathbf{b}'||^2}MNU\big((\sum_{i=1}^{N}-f'_i\hat{t}_i)^2+(\sum_{i=1}^{N}f'_i\hat{r}_i)^2\big)},
\label{crbsimplified}
\end{align}
where $\hat{R}, \hat{T},$ $U$, $\hat{r}_i$, $\hat{t}_i$, and $u_i$ are defined in Subsection \ref{opuserlb} and $Q=\frac{\sigma^2_r}{2 |\alpha| ^2 L}$, $||\mathbf{b}'||^2=\frac{\pi^2\cos^2(\theta)M(M^2-1)}{12}$ and $\psi \triangleq |c_1|^2\big(\hat{R}^2+\hat{T}^2\big)+|c_2|^2N U$.
\end{lemma}
\begin{proof}
The proof is provided in Appendix \ref{lemma4p}.
\end{proof}
Deriving the sensing OP here is more challenging than in the SJB analysis. This is because the CRB expression for LB involves a complex rational function where both the numerator and denominator contain products and sums of correlated random variables, including quadratic forms involving index-dependent terms \(\hat{f}_i\) and \(\hat{r}_i, \hat{t}_i\) that depend on the channel realizations, making a direct analytical derivation of its probability density function intractable. Therefore, we derive upper and lower bounds, as well as a computationally efficient approximation, to characterize the sensing OP for LB.
\begin{proposition}\label{upperbound}
An upper bound for the sensing OP of LB is given by $P_{U}(\epsilon) = I(\check{f}(\hat{R}, \hat{T},U), \frac{\frac{Q}{||\mathbf{b}'||^2}U}{\psi} >\epsilon)$. 
\end{proposition}
\begin{proof}
If we neglect the term $\big((\sum_{i=1}^{N}-f'_i\hat{t}_i)^2+(\sum_{i=1}^{N}f'_i\hat{r}_i)^2\big)$ in the denominator of the expression (\ref{crbsimplified}) for CRB($\theta$), we obtain an upper bound on the CRB, which we denote as UCRB. Since CRB($\theta$)$<$UCRB($\theta$),
$P_{\text{out}, s}(\epsilon)=P(\text{CRB}(\theta)>\epsilon)<P(\text{UCRB}(\theta)\!\!>\epsilon)\triangleq P_{U}(\epsilon)$,  i.e., $P_{U}$ provides an upper bound on the sensing OP. Moreover, following a similar approach as in Subsection \ref{opuserlb} and using the result of Lemma \ref{lemma5i}, the proof is complete.
\end{proof}
\begin{proposition}\label{lowerbound}
A lower bound for the sensing OP of LB is given by $P_{L}(\epsilon) = 
I(\check{f}(\hat{R}, \hat{T},U), \frac{\frac{Q}{||\mathbf{b}'||^2}U \psi}{\psi^2+\!\frac{|c_1c_2|^2}{||\mathbf{b}'||^2}MNU^2(\sum_{i=1}^{N}\!|f'_i|^2)}>\epsilon)$. 
\end{proposition}
\begin{proof}
Using the Cauchy–Schwarz inequality, we have
\begin{align}
&\big((\sum_{i=1}^{N}-f'_i\hat{t}_i)^2+(\sum_{i=1}^{N}f'_i\hat{r}_i)^2\big)=|\sum_{i=1}^{N}\!jf'_ie^{jf_i}\tilde{h}_i|^2\nonumber\\
&<(\sum_{i=1}^{N}\!|jf'_i|^2)(\sum_{i=1}^{N}\!|e^{jf_i}\tilde{h}_i|^2)=(\sum_{i=1}^{N}\!|f'_i|^2)U,
\end{align}
where we have used $\sum_{i=1}^{N}\!|f'_i|^2=\frac{\pi^2\cos^2(\theta)N(N^2-1)}{12}$. By replacing $\big((\sum_{i=1}^{N}-f'_i\hat{t}_i)^2+(\sum_{i=1}^{N}f'_i\hat{r}_i)^2\big)$ with $(\sum_{i=1}^{N}\!|f'_i|^2)U$ in the denominator of the expression (\ref{crbsimplified}) for CRB($\theta$), we obtain a lower bound on the CRB, which we denote as LCRB. Thus, $P_{\text{out}, s}(\epsilon)\!=\!\!P(\text{CRB}(\theta)\!>\epsilon)\!>\!P(\text{LCRB}(\theta)\!>\!\epsilon)\!\triangleq P_{L}(\epsilon)$, where $P_{L}$ is a lower bound on the sensing OP. Then, following a similar approach as in Subsection \ref{opuserlb} and using the result of Lemma \ref{lemma5i}, the proof is complete.
\end{proof}
\begin{proposition}\label{approximation}
An approximation of the sensing OP of LB is given by $P_{A}(\epsilon) = I(\check{f}(\hat{R}, \hat{T},U), \frac{\frac{Q}{||\mathbf{b}'||^2}U \psi}{\psi^2+\!\frac{|c_1c_2|^2}{||\mathbf{b}'||^2}MNU\tilde{A}}>\epsilon)$, when $N$ is large.
\end{proposition}
\begin{proof}
The proof is provided in Appendix \ref{lemmafp}.
\end{proof}
In Section \ref{simulations}, we illustrate through numerical results that this approximation closely aligns with the exact sensing-outage curve for all $\epsilon$ considered.

\textbf{Special Cases of Interest:}\label{simplified4}
In the following, we consider some special cases which result in a simplified final expression for the approximation of $P_{\text{out}, s}(\epsilon)$.

1) The case $K=0$ corresponds to Rayleigh fading channels to the user. In this case, the result reduces to that of Subsection B.3 in \cite{OnStochasticFundamentalLimitsina}.

2) The case $K \rightarrow \infty$ corresponds to pure LoS channels to the user. In this case, after some algebraic calculation, we have \(U = N\), \(\hat{R} =\frac{\sin(N\Delta)}{\sin\Delta}\), \(\hat{T} = 0\), and \(\tilde{A} = \left[\left(\sum_i f_i'\cos(f_i-f^u_i)\right)^2 + \left(\sum_i f_i'\sin(f_i-f^u_i)\right)^2\right]\). Thus, $P_{\mathrm{out},s}(\epsilon)=\iint_{\mathcal {\check{D}}_2(\epsilon)}
f_{\theta,\theta_u}(\theta,\theta_u)
\mathrm{d}\theta_u
\mathrm{d}\theta$, where $\mathcal {\check{D}}_2(\epsilon)
= \{ (\theta,\theta_u) \in [0,\pi]^2 :
\frac{\frac{Q}{||\mathbf{b}'||^2}N \psi}{\psi^2+\!\frac{|c_1c_2|^2}{||\mathbf{b}'||^2}MN^2\tilde{A}} > \epsilon \}$ and \(\psi = |c_1|^2 (\frac{\sin(N\Delta)}{\sin\Delta})^2 + |c_2|^2 N^2\). If either the target angle or the user angle is fixed, the integral becomes a 1-D integral over the domain of the remaining angle.

3) If \(\theta_u = \theta\) (\(f_i = f^u_i\)), then \(\Delta = 0\) and all \(\tilde{\mathbf{d}}_i\) become i.i.d. Moreover, we have $\tilde{A} = \alpha_2^2 \cdot \frac{\pi^2 \cos^2(\theta)\, N(N^2 - 1)}{12}.$ Then,
$P_{\text{out}, s}(\epsilon) = P\left(\text{ACRB}(\theta) > \epsilon\right)
= \int_{0}^{\pi} P\left(\text{ACRB}(\theta) > \epsilon \mid \theta\right) f_\theta(\theta) \mathrm{d}\theta$.
\section{Opportunistic Sensing/Communication Approaches}\label{specialpoints}
In this section, we derive the sensing OP and the communication OP at boundary points of the SJB and LB approaches. These points represent scenarios where either: 1) Sensing performance dominates, termed ``opportunistic communication," achieved by setting $b_1=0$ and $(c_1, c_2) = (0,\sqrt{p_t})$ for the case of SJB and LB, respectively, or 2) Communication performance dominates, termed ``opportunistic sensing," achieved by setting $b_1=\infty$ and $(c_1, c_2) = (\sqrt{p_t},0)$ for the case of SJB and LB, respectively.
\subsection{Opportunistic Points for the SJB Approach}\label{opportunistic}
First, we analyze the opportunistic communication point. Based on (\ref{sinr2i}) and the definitions of $\hat{R}$ and $\hat{T}$ in Subsection \ref{opuserlb}, for the case $b_1=0$ the SINR of the user is given by $\frac{p_t}{\sigma^2_u}\frac{\hat{R}^2 + \hat{T}^2}{N}$. Moreover, based on Lemma \ref{lemma5i}, we have
$[\hat{R}, \hat{T}]^T \stackrel{\text{d}}{\approx} \mathcal{N}\left(
\begin{bmatrix}
\mu_{\hat{R}} \\ \mu_{\hat{T}}
\end{bmatrix},
\begin{bmatrix}
\frac{N\alpha_2^2}{2} & 0 \\
0 & \frac{N\alpha_2^2}{2}
\end{bmatrix}
\right)
$ where $\mu_{\hat{R}} = \sum_{i=1}^N \alpha_1\cos(f_i-f^u_i)$ and $\mu_{\hat{T}} = \sum_{i=1}^N \alpha_1\sin(f_i-f^u_i)$. Thus, $P( \frac{p_t}{\sigma_u^2} \cdot \frac{\hat{R}^2 + \hat{T}^2}{N} < \gamma |\theta, \theta_u)
= P( (( \frac{\hat{R}}{\sqrt{N\alpha_2^2/2}})^2 + (\frac{\hat{T}}{\sqrt{N\alpha_2^2/2}})^2) < 2\frac{\gamma\sigma_u^2}{p_t\alpha_2^2}|\theta,\theta_u)= F_{\chi^2_2(\lambda)}( \frac{2\gamma\sigma_u^2}{p_t\alpha_2^2})$, where \( F_{\chi^2_2(\lambda)}(x) \) is the CDF of the non-central chi-squared distribution with 2 degrees of freedom and non-centrality parameter $\lambda=\frac{\mu_{\hat{R}}^2 + \mu_{\hat{T}}^2}{N\alpha_2^2/2}= \frac{2}{N\alpha_2^2} \left| \sum_{i=1}^N \alpha_1 e^{j(f_i - f^u_i)} \right|^2= \frac{2\alpha_1^2}{N\alpha_2^2} \cdot \frac{\sin^2\left( \frac{\pi N (\sin\theta - \sin\theta_u)}{2} \right)}{\sin^2\left( \frac{\pi (\sin\theta - \sin\theta_u)}{2} \right)}$. The term \( \frac{\sin^2(\pi N x/2)}{\sin^2(\pi x/2)} \) peaks at \( x=0 \) (user and target in a similar direction) and has nulls when \( \pi N x/2 = k\pi \), i.e., \( \sin\theta - \sin\theta_u = \frac{2k}{N} \). Thus, $P( \frac{p_t}{\sigma_u^2} \cdot \frac{\hat{R}^2 + \hat{T}^2}{N} < \gamma)=\iint_{0}^{\pi}F_{\hat{A}}( \frac{2\gamma\sigma_u^2}{p_t\alpha_2^2})f_{\theta,\theta_u}(\theta,\theta_u)\mathrm{d}\theta \mathrm{d}\theta_u,$ where $\hat{A}=\chi^2_2(\frac{2\alpha_1^2}{N\alpha_2^2} \cdot \frac{\sin^2\left( \frac{\pi N (\sin\theta - \sin\theta_u)}{2} \right)}{\sin^2\left( \frac{\pi (\sin\theta - \sin\theta_u)}{2} \right)})$

Moreover, at the opportunistic communication point, based on (\ref{cramer}), the CRB is equal to $\frac{g(\theta)}{N}$. Thus, the sensing OP is
\begin{align}
P_{\text{out}, s}(\epsilon)\!\!&=\!\!P(\text{CRB}(\theta)>\epsilon)\!\!= \iint\limits_{ \theta \in [0, \pi] : g(\theta) > N\epsilon }f_{\theta,\theta_u}(\theta,\theta_u)\mathrm{d}\theta \mathrm{d}\theta_u.\label{crbb1zero}
\end{align}
Next, we analyze the opportunistic sensing point. As \(b_1\) approaches infinity, based on (\ref{sinr2i}) and the definition of $U$ in Subsection \ref{opuserlb}, the SINR of the user approaches \(\frac{p_t U|c|^2}{(\sigma_u^2)}\). Thus, using the approach of Subsection \ref{opuserlb}, we obtain $P\left( U < \frac{\Gamma\sigma^2_u}{p_t|c|^2} \right)=\Phi\left( \frac{\frac{\Gamma\sigma^2_u}{p_t|c|^2} - N}{\sqrt{N(\alpha_2^4 + 2\alpha_1^2\alpha_2^2)}} \right)$.

Furthermore, at the opportunistic sensing point, based on (\ref{cramer}) and the definitions of $\hat{R}$, $\hat{T}$, and $U$ in Subsection~\ref{opuserlb}, the CRB is \(\frac{|\mathbf{\tilde{h}}|^2 g(\theta)}{|\mathbf{a}^H\mathbf{\tilde{h}}|^2} = \frac{g(\theta) {U}}{\hat{R}^2 + \hat{T}^2}\). Thus, following a similar approach as in Subsection \ref{opuserlb}, the sensing OP is $P_{\text{out}, s}(\epsilon) = I(f(\hat{R}, \hat{T},U), \frac{g(\theta) {U}}{\hat{R}^2 + \hat{T}^2}>\epsilon)$ where $\tilde{\hat{\boldsymbol{\mu}}}_i$ and $\tilde{\hat{\boldsymbol{\Sigma}}}_i$ are given in Lemma \ref{lemma5i}.
\subsection{Opportunistic Points for the LB Approach}
First, we analyze the opportunistic communication point. Based on (\ref{SINRnotdpc}), the SINR of the user for the case $|c_1|=0$ is zero, and the communication OP is equal to one. Moreover, in this case $(c_1,c_2)=(0,\sqrt{p_t})$, based on (\ref{crbsimplified}), the CRB is equal to $\frac{Q}{N p_t||\mathbf{b}'||^2}$. By substituting $Q$ and $||\mathbf{b}'||^2$ from Lemma \ref{lemma4} we have $\frac{Q}{N p_t||\mathbf{b}'||^2}=\frac{g(\theta)}{N}$, which is the same as CRB of the SJB approach when $b_1=0$. Thus, the sensing OP is the same as in (\ref{crbb1zero}).

Next, we analyze the opportunistic sensing point. When $(c_1,c_2)=(\sqrt{p_t},0)$, based on (\ref{SINRnotdpc}), the SINR of the user will be \(\frac{p_t c^2U}{(\sigma_u^2)}\), which is the same as the SINR for the user in the SJB case when \(b_1=\infty\). Thus, the communication OP for the case of LB is the same as the communication OP for the case of SJB at the opportunistic sensing point. Moreover, based on (\ref{crbsimplified}), the CRB is equal to \(\frac{Q}{p_t ||\mathbf{b}'||^2} \frac{{U}}{\hat{R}^2 + \hat{T}^2}\). Since \(\frac{Q}{p_t ||\mathbf{b}'||^2} = g(\theta)\), the sensing OP is the same as in the SJB scenario when $b_1$ tends to infinity, derived in Subsection \ref{opportunistic}.
\section{Asymptotic Scaling Laws}\label{asympt}
In this section, we investigate the asymptotic behavior of the considered ISAC system to reveal the fundamental limits under each of the beamforming schemes. While performance naturally improves with increasing antennas and power, the key aim of this section is to identify which schemes are power-limited 
%\footnote{Power-limited behavior means the communication SINR increases with available resources (antennas $N$ or transmit power $p_t$). Adding antennas or power directly raises SINR, enabling arbitrarily high rate thresholds $\gamma$. Performance is limited only by resources, with no fundamental ceiling.} 
versus interference-limited 
%\footnote{Interference-limited behavior means the communication SINR saturates at a constant value, regardless of added antennas or transmit power. This happens because the radar sensing signal leaks into the user’s receiver, creating self-interference that scales with the desired signal. The limiting factor is this interference floor, not the available resources.} 
and quantifying the exact scaling rates.
\subsection{Large-System Scaling (\(N \to \infty\))}\label{largeN}
As the number of transmit antennas \(N\) increases, channel hardening occurs, but the limiting behavior differs markedly across beamforming schemes.
\\
\textbf{SJB:} We begin by establishing the convergence of the key random variables in the SJB scheme.
\begin{lemma}\label{SJBuserNinf}
As \(N \to \infty\), the normalized sums \(\frac{X}{N}\), \(\frac{Y}{N}\), and \(\frac{Z}{N}\) converge almost surely to deterministic constants. Consequently, $\mathrm{SINR} = \frac{p_t c^2}{\sigma_u^2} \cdot \frac{X^2 + Y^2}{Z} \xrightarrow{\text{a.s.}} \mathrm{SINR}_\infty^{\mathrm{(SJB)}}=\frac{p_t c^2}{\sigma_u^2}\frac{|b_1|^2}{|b_1|^2 + |b_2|^2}$.
\end{lemma}
\begin{proof}
The proof is provided in Appendix \ref{ProofSJBuserNinf}.
\end{proof}
\textbf{Note:} In the special case where \(\sin\Delta = 0\) (i.e., \(\sin\theta_u = \sin\theta\)), we have \(\frac{\sin(N\Delta)}{\sin\Delta} = N\), and the limits become
$\mathrm{SINR}_\infty^{\mathrm{(SJB)}} = \frac{p_t c^2}{\sigma_u^2} \cdot \frac{|b_1|^2 + |b_2|^2\alpha_1^2 + 2|b_1b_2|\alpha_1\cos(\phi_1 - \phi_2)}{|b_1|^2 + |b_2|^2 + 2\alpha_1|b_1b_2|\cos(\phi_1 - \phi_2)}$. However, this is a measure-zero event and does not affect the general analysis.

This result demonstrates channel hardening: the random SINR becomes deterministic in the large-\(N\) limit. Therefore, for any fixed communication threshold \(\gamma\), the OP converges to a step function $P_{\mathrm{out},c}(\gamma) \to 
\begin{cases}
0 & \text{if } \gamma < \mathrm{SINR}_\infty^{\text{(SJB)}} \\
1 & \text{if } \gamma > \mathrm{SINR}_\infty^{\text{(SJB)}}
\end{cases}$. For sensing, following a similar approach as in Lemma \ref{SJBuserNinf}, we can show that \(\tilde{X}/N\), \(\tilde{Y}/N\), and \(Z/N\) converge to constants. Substituting into (\ref{cramer}) yields $\mathrm{CRB}(\theta) = \frac{g(\theta) Z}{\tilde{X}^2 + \tilde{Y}^2} = \mathcal{O}\left(\frac{1}{N}\right)$.  Thus, for any fixed sensing threshold \(\epsilon > 0\), there exists \(N_0\) such that for all \(N > N_0\), \(\mathrm{CRB}(\theta) < \epsilon\) almost surely. Hence $P_{\mathrm{out},s}(\epsilon) = P(\mathrm{CRB}(\theta) > \epsilon) \to 0 \quad \text{as } N \to \infty$.

\textbf{LB:} For the communication OP in LB without DPC, we have the following Lemma.
\begin{lemma}\label{LBuserNinf}
Consider LB without DPC and condition on a fixed angle pair $(\theta,\theta_u)$, with
$
\Delta=\frac{\pi}{2}\big[\sin(\theta_u)-\sin(\theta)\big].
$
Assume the Rician factor $K$ is fixed as $N\to\infty$.

1) If $\sin(\Delta)\neq 0$, then
$
\mathrm{SINR}=\Theta_p(N),
$
and hence, for every fixed $\gamma>0$,
$
P_{\mathrm{out},c}(\gamma\mid \theta,\theta_u)\to 0.
$

2) If $\Delta=0$ and $K>0$, then
$
\mathrm{SINR}\xrightarrow{p}\frac{|c_1|^2}{|c_2|^2\alpha_1^2},
$
and therefore
$
P_{\mathrm{out},c}(\gamma\mid \theta,\theta_u)\to
\begin{cases}
0, & \gamma<\dfrac{|c_1|^2}{|c_2|^2\alpha_1^2},\\[1ex]
1, & \gamma>\dfrac{|c_1|^2}{|c_2|^2\alpha_1^2}.
\end{cases}
$
\end{lemma}
\begin{proof}
Define
$
\hat H_N \triangleq \hat R+j\hat T
= \sum_{i=1}^N \left(\alpha_1 e^{j(f_i-f_i^u)}+\alpha_2 e^{j f_i}h_i\right)
= \alpha_1 S_N+\alpha_2 V_N,
$
where
$
S_N\triangleq \sum_{i=1}^N e^{j(f_i-f_i^u)},
\qquad
V_N\triangleq \sum_{i=1}^N e^{j f_i}h_i.
$
Since multiplication by $e^{j f_i}$ only rotates $h_i$, the random variables $e^{j f_i}h_i$ remain i.i.d. circularly symmetric with zero mean and unit variance. Hence
$
V_N = O_p(\sqrt{N}).
$
Also, by the strong law of large numbers,
$
\frac{U}{N}\xrightarrow{\mathrm{a.s.}}1.
$

If $\sin(\Delta)\neq 0$, then
$
S_N=e^{j\Delta(N-1)}\frac{\sin(N\Delta)}{\sin(\Delta)},
$
so $S_N=O(1)$. Therefore
$
\hat H_N=\alpha_1 O(1)+\alpha_2 O_p(\sqrt{N})=O_p(\sqrt{N}),
$
and thus
$
\hat R^2+\hat T^2 = |\hat H_N|^2 = O_p(N).
$
Substituting into \eqref{SINRnotdpc},
$
\mathrm{SINR}
=
\frac{|c_1|^2U}{\frac{\sigma_u^2}{|c|^2}+\frac{|c_2|^2}{N}(\hat R^2+\hat T^2)}
=
\frac{|c_1|^2\,\Theta(N)}{\Theta_p(1)}
=\Theta_p(N).
$
Hence $P_{\mathrm{out},c}(\gamma\mid\theta,\theta_u)\to 0$ for every fixed $\gamma>0$.

If $\Delta=0$, then $S_N=N$, so
$
\hat H_N=\alpha_1 N+\alpha_2 V_N.
$
Since $V_N/N\xrightarrow{p}0$, we obtain
$
\frac{\hat H_N}{N}\xrightarrow{p}\alpha_1,
\qquad
\frac{\hat R^2+\hat T^2}{N^2}=\frac{|\hat H_N|^2}{N^2}\xrightarrow{p}\alpha_1^2.
$
Using also $U/N\xrightarrow{\mathrm{a.s.}}1$, \eqref{SINRnotdpc} gives
$
\mathrm{SINR}
=
\frac{|c_1|^2U}{\frac{\sigma_u^2}{|c|^2}+\frac{|c_2|^2}{N}(\hat R^2+\hat T^2)}
\xrightarrow{p}
\frac{|c_1|^2}{|c_2|^2\alpha_1^2}.
$
This completes the proof.
\end{proof}
Thus, unlike SJB, the large-$N$ communication behavior of LB without DPC depends on the angular alignment regime. In the generic non-aligned case, the OP vanishes for every fixed threshold $\gamma$, whereas in the aligned case it converges to a step function determined by $\frac{|c_1|^2}{|c_2|^2\alpha_1^2}$. Since $\Delta=0$ is a measure-zero event when $(\theta,\theta_u)$ has a continuous joint distribution, the unconditional communication OP still vanishes as $N\to\infty$ at LB.

For sensing in LB, define
$
B_N \triangleq \left(\sum_{i=1}^{N}-f'_i\hat{t}_i\right)^2+\left(\sum_{i=1}^{N}f'_i\hat{r}_i\right)^2
= \left|\sum_{i=1}^{N} j f'_i e^{j f_i}\tilde h_i\right|^2.
$
For fixed finite $K$, the random NLoS term dominates the growth of $B_N$, and since
$
\sum_{i=1}^{N}|f'_i|^2=\Theta(N^3),
$
we obtain
$
B_N=\Theta_p(N^3).
$
Moreover, $U=\Theta(N)$ and
$
\psi=|c_1|^2(\hat R^2+\hat T^2)+|c_2|^2NU=\Theta_p(N^2).
$
Substituting these scalings into \eqref{crbsimplified} yields
$
\mathrm{CRB}(\theta)=\Theta_p\!\left(\frac{1}{N^2}\right).
$
Hence, for any fixed sensing threshold $\epsilon>0$,
$
P_{\mathrm{out},s}(\epsilon)\to 0
\qquad \text{as } N\to\infty.
$
\subsection{High-Power Scaling (\(p_t \to \infty\))}\label{largePt}
As the transmit power increases, the scaling behavior further distinguishes the beamforming schemes.
\\
\textbf{SJB and LB with DPC:} From (\ref{sinr2i}) and Lemma \ref{lemma1i}, \(\mathrm{SINR} \propto p_t\) for SJB. From Section \ref{opuserlb}, the same holds for LB with DPC. Therefore, for any fixed \(\gamma\): $P_{\mathrm{out},c}(\gamma) \to 0$. For sensing, (\ref{cramer}) and Lemma \ref{lemma4} show that \(g(\theta) \propto 1/p_t\). Hence \(\mathrm{CRB} = \mathcal{O}(1/p_t)\) and $P_{\mathrm{out},s}(\epsilon) \to 0$ as $p_t \rightarrow \infty.$
\\
\textbf{LB without DPC:} From (\ref{SINRnotdpc}), with \(|c_1|^2 + |c_2|^2 = p_t\), both the numerator \(|c_1|^2 U\) and the denominator \(\frac{|c_2|^2}{N}(\hat{R}^2+\hat{T}^2)\) scale as \(\mathcal{O}(p_t)\). Therefore, the SINR approaches a positive constant value as $p_t \rightarrow \infty$. For the sensing, we have \(\mathrm{CRB} = \mathcal{O}(1/p_t)\) and \(P_{\mathrm{out},s}(\epsilon) \to 0\) as $p_t \rightarrow \infty.$
\begin{figure*}
    \centering
    \subfigure[Communication performance, SJB]{\includegraphics[scale=.44]{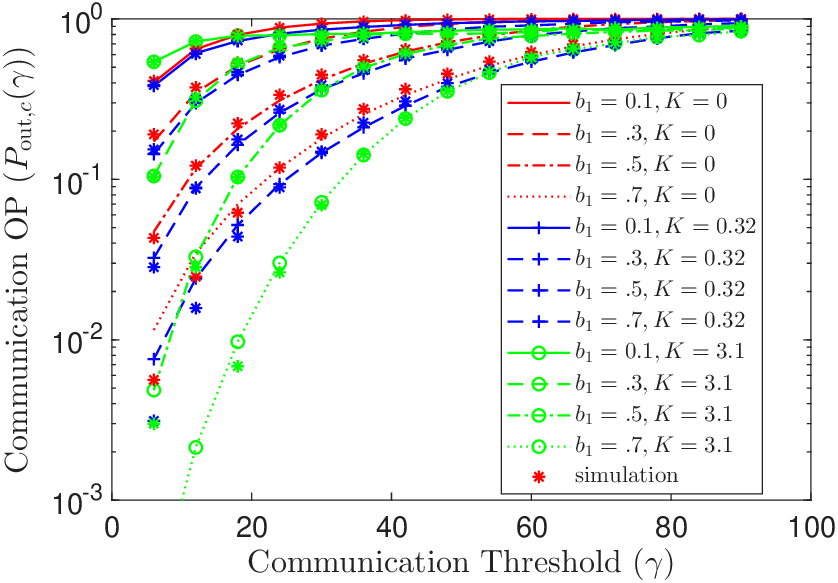}}\label{4}
    \subfigure[Communication performance, LB]{\includegraphics[scale=.44]{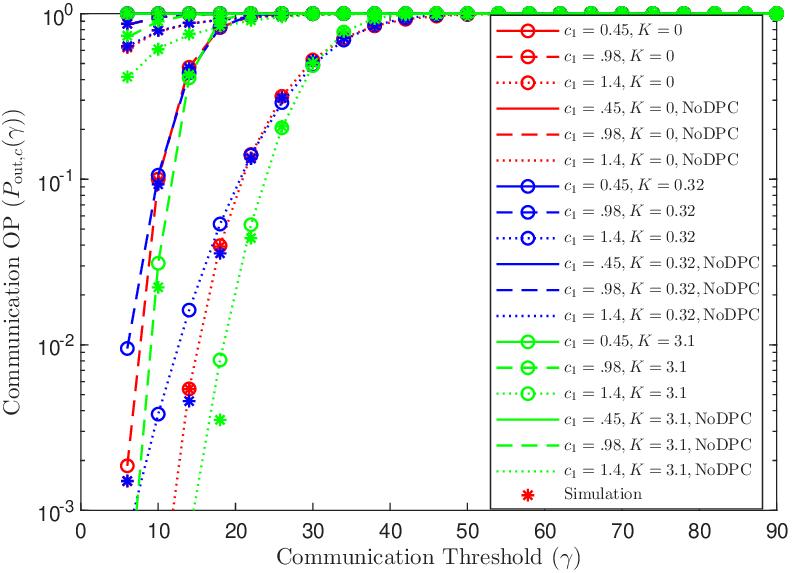}}\label{5}
    \subfigure[Sensing performance, SJB]{\includegraphics[scale=.4]{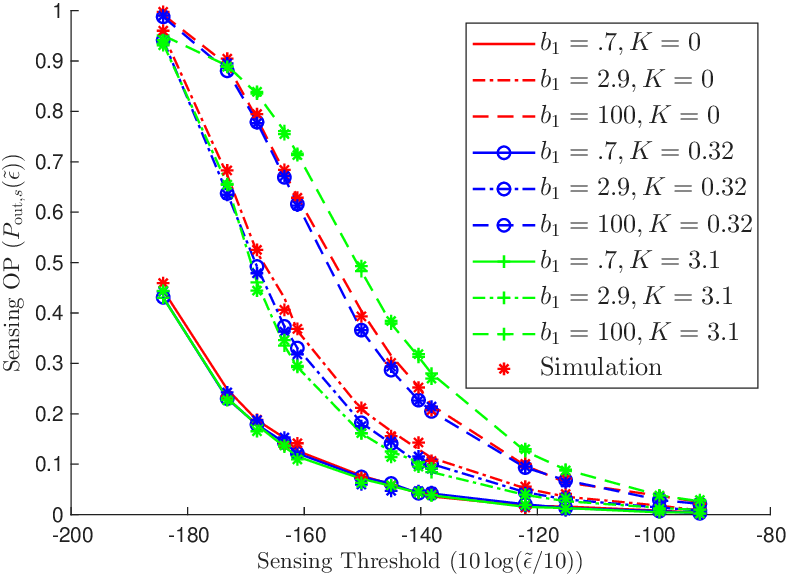}}\hspace{-0.5cm}\label{5}
    \caption{Communication OP ($P_{\mathrm{out},c}(\gamma)$) versus communication threshold ($\gamma$) at SJB, communication OP versus communication threshold at LB, and sensing OP $P_{\mathrm{out},s}(\epsilon)$ versus sensing threshold ($\epsilon$) (log scale) at SJB.}
    \label{opuser}
\end{figure*}
\subsection{Summary of Asymptotic Scaling Laws}
We summarize the fundamental asymptotic behavior of each scheme as follows. For communication, SJB is power-limited, with $\mathrm{SINR} \propto p_t$, and exhibits channel hardening, where the random SINR converges to a deterministic constant as $N \to \infty$. LB with DPC is power-limited in both $N$ and $p_t$, achieving $\mathrm{SINR} \sim \mathcal{O}(N)$ and $\mathrm{SINR} \propto p_t$, thereby fully exploiting the available resources. In contrast, LB without DPC is interference-limited in the high-power regime $p_t\to\infty$, while its large-$N$ communication behavior depends on angular alignment; for generic non-aligned user/target angle pairs the outage probability vanishes as $N\to\infty$. For sensing, all schemes are power-limited, with $\mathrm{CRB} = \mathcal{O}(1/p_t)$. The scaling with the number of antennas, however, differs markedly: the CRB scales as $\mathcal{O}(1/N)$ for SJB, whereas for LB it scales as $\Theta_p(1/N^2)$, yielding a faster asymptotic improvement in sensing accuracy.
\section{Numerical Results}\label{simulations}
Unless stated otherwise, we consider the following parameters. The BS is equipped with \(N = 15\) transmit antennas and \(M = 17\) receive antennas. The total transmit power budget is $p_t = 10$. The noise variances at the user's receiver and at the radar receiver are $\sigma_u^2 = 1$ and $\sigma_r^2 = 1$, respectively. The length of the radar/communication frame is \(L = 30\) symbols. The beamforming parameters have phases \(\phi_1 = \frac{\pi}{3}\) rad and \(\phi_2 = 0\) rad. For the case of SJB, we set \(b_2 = 1\) and vary \(b_1\) from \(0\) to \(\infty\). In the LB case, the power allocated to the communication beamformer, \(|c_1|^2\), can take any value between \(0\) and \(p_t\). The target reflection coefficient is \(\alpha = 1\), combining the radar cross section and path loss effects \footnote{As shown in (\ref{crb}), $\text{CRB}(\theta,\alpha,\sigma_r)$ is a function of $\alpha$, $\sigma_r$, and $\theta$. By defining $\tilde{\epsilon}(\alpha) = \frac{\epsilon |\alpha|^2}{\sigma_r^2}$, where $\alpha$ and $\sigma_r$ denote the reflection coefficient and the noise variance of the echo signal at BS, respectively, the sensing OP can be expressed as
$P_{\text{out}, s}(\epsilon) = P(\text{CRB}(\theta,\alpha,\sigma_r) > \epsilon) = P(\text{CRB}(\theta,\alpha=1,\sigma_r=1) > \tilde{\epsilon})$. Therefore, since $\tilde{\epsilon}(\alpha) = \frac{\epsilon |\alpha|^2}{\sigma_r^2}$, different values of $\tilde{\epsilon}$ in the simulations can be interpreted as corresponding to either different values of $\epsilon$ or different values of $\alpha$.}. We assume that the target and user angles are independent and uniformly distributed over the interval $[0, \pi]$. However, our numerical derivations are valid for any arbitrary joint distribution of the angles. Three different values of the Rician factor $K$ are considered: $K = 0$, corresponding to Rayleigh fading; $K = 0.32$, corresponding to moderate LoS; and $K = 3.1$, corresponding to strong LoS.

Fig.\ref{opuser}(a) and Fig.\ref{opuser}(b) show the communication OP for SJB and LB, respectively, as a function of the SINR threshold \(\gamma\), for different values of the Rician factor \(K\) and beamforming parameters. In all cases, the Monte Carlo simulations closely match the analytical results, validating our derivations. As expected, increasing \(\gamma\) raises the OP for both schemes. Allocating more power to the communication beamformer, i.e., increasing \(|b_1|\) in SJB or \(|c_1|\) in LB, reduces the OP, since the beam is steered more strongly toward the user's direction. For SJB, the improvement is most pronounced for small \(|b_1|\). Fig.\ref{opuser}.(b) reveals that without DPC, LB suffers from severe self-interference caused by the radar signal. Even at low \(\gamma\), the OP remains high, and communication reliability is severely limited. In contrast, with DPC, the radar interference is effectively mitigated, leading to much lower OP across all values of the threshold \(\gamma\). This shows that, under the considered LB design and parameter setting, DPC is highly beneficial for realizing the communication potential of LB. The influence of \(K\) differs markedly between SJB and LB, and also depends on whether DPC is employed. 

In SJB (Fig.~\ref{opuser}(a)), the impact of the Rician $K$-factor is regime-dependent. Moderate LoS consistently yields lower OP than Rayleigh fading across all $b_1$ and $\gamma$,  indicating that even a modest deterministic component improves reliability by enabling partial exploitation of the channel structure. Strong LoS exhibits more nuanced behavior: at low beamforming weight ($b_1 = 0.1$), it initially underperforms at small $\gamma$ because the dominant path may misalign with the weak beamformer, reducing effective diversity compared to fully random scattering. As $\gamma$ increases, however, strong LoS achieves the lowest OP. As $b_1$ grows, strong LoS becomes superior at low-to-moderate $\gamma$, while moderate LoS can regain advantage at very high thresholds. At moderate OP levels (e.g., $P_{\text{out},c} \approx 0.1$), strong LoS consistently yields the minimum OP for all $b_1$.

In LB (Fig.~\ref{opuser}(b)), the impact of the Rician $K$-factor depends critically on whether DPC is employed. With DPC, the effect is non-monotonic: increasing $K$ from $0$ to $0.32$ slightly increases $P_{\text{out},c}$ for most $\gamma$, suggesting that a weak LoS component can degrade performance, possibly due to phase mismatch affecting interference cancellation. Increasing $K$ further to $3.1$, however, significantly reduces $P_{\text{out},c}$ at low-to-moderate $\gamma$ as the strong LoS path provides stable gain. At very high SINR thresholds, this advantage diminishes and strong LoS can even exhibit higher OP than moderate LoS. At moderate outage levels (e.g., $P_{\text{out},c} \approx 0.1$), strong LoS consistently achieves the lowest OP for all $c_1$, confirming its benefit in practical scenarios. Without DPC, the trend is monotonic: increasing $K$ uniformly decreases $P_{\text{out},c}$ across all $\gamma$ and $c_1$, as a stronger LoS enhances the desired signal while the self-interference structure remains unchanged. Nevertheless, OP without DPC remains substantially higher than with DPC, underscoring the necessity of interference cancellation.

Fig.~\ref{opuser}(c) shows the sensing OP of SJB versus the log-scale threshold $\tilde{\epsilon}$ (scaling with $\alpha$ and $\epsilon$), for different $b_1$ and $K$. As $\tilde{\epsilon}$ increases, the sensing OP decreases exponentially, reflecting relaxed accuracy requirements or stronger reflections. Increasing $b_1$ shifts power from sensing to communication, raising the sensing OP, most significantly at small to moderate $b_1$; beyond $b_1\approx10$, further increases have little effect as the beamformer becomes user-focused. The impact of $K$ is regime-dependent. Moderate LoS yields consistently lower sensing OP than Rayleigh fading across all $b_1$ and $\tilde{\epsilon}$, showing that even a modest LoS component improves angle estimation accuracy. Strong LoS exhibits contrasting behavior: at low to moderate $b_1$ (e.g., $0.7$ and $2.9$), it achieves the lowest sensing OP among all $K$ values, indicating that a dominant LoS enhances sensing when sufficient power is directed toward the target. At high $b_1$ (e.g., $100$), however, strong LoS yields the highest sensing OP, a very strong LoS becomes detrimental when communication is prioritized, likely due to increased directional imbalance.
\begin{figure}
    \centering
    \subfigure[Different $c_1$, $K=3.1$]{\includegraphics[scale=.44]{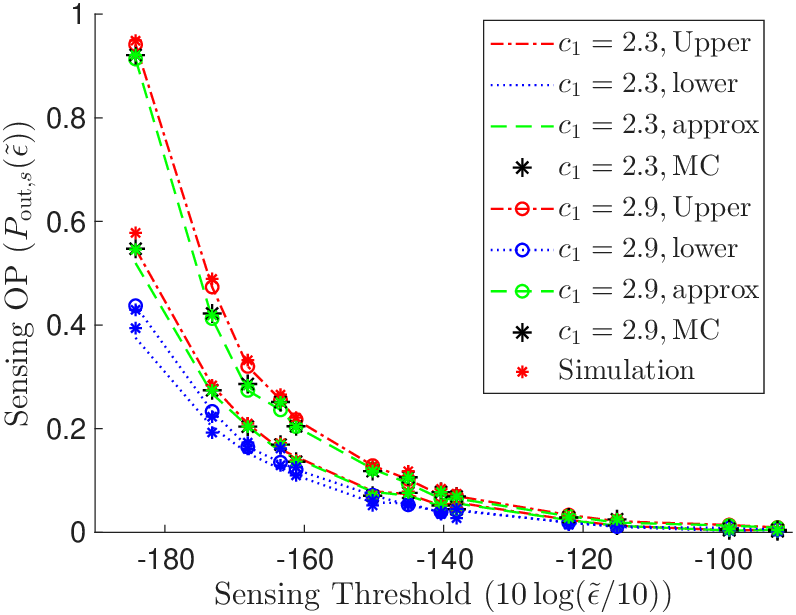}}\hspace{-0.5cm}\label{10}
    \subfigure[Different Rician factor, $K$]{\includegraphics[scale=.44]{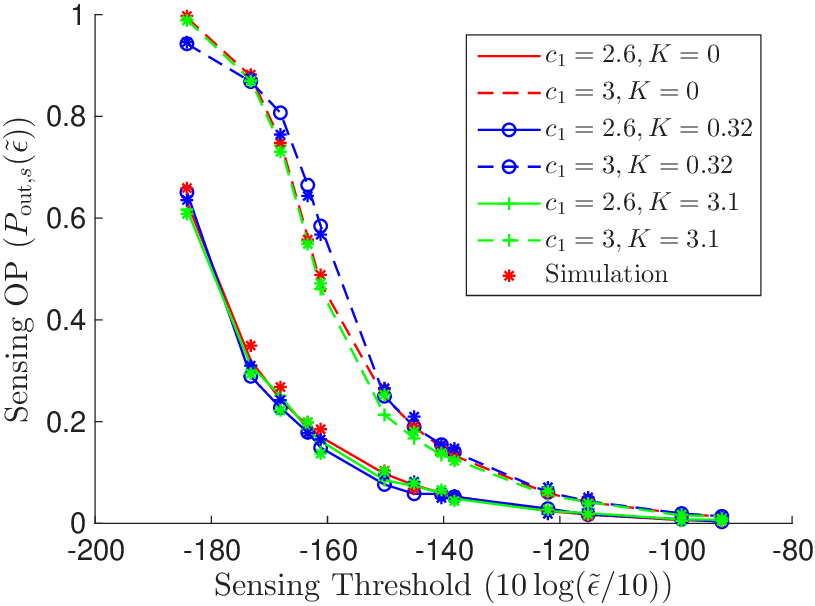}}\hspace{-0.5cm}\label{8}
    \subfigure[$K=3.1$, $c_1=2.6$, and different ($N, M, p_t$)]{\includegraphics[scale=.44]{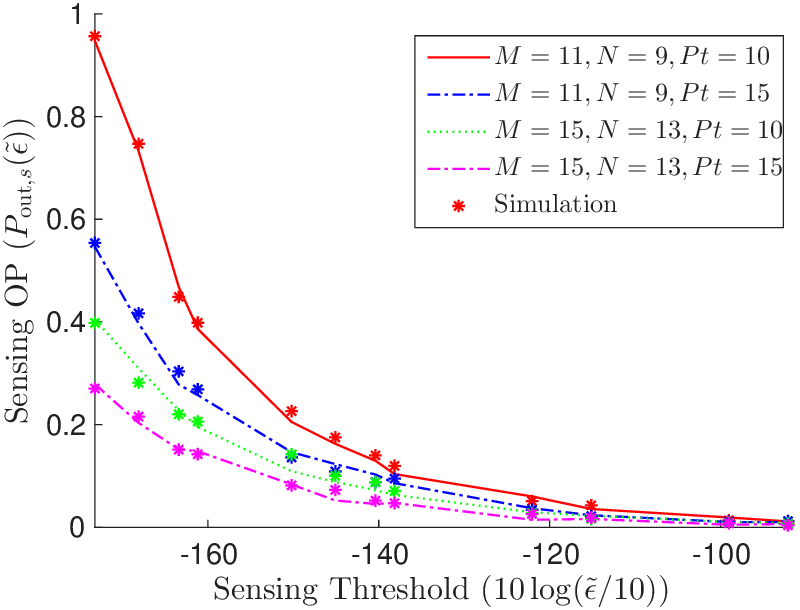}}\hspace{-0.5cm}\label{7}
    \caption{Sensing OP at LB}
    \label{optarget}
\end{figure}

In Fig.~\ref{optarget}(a) and Fig.~\ref{optarget}(b), the sensing OP for LB is shown for different $c_1$ versus the threshold $\tilde{\epsilon}$. In the legend, "MC main" refers to the Monte Carlo simulation where the CRB evaluation is based on Lemma \ref{lemma4}. As $\tilde{\epsilon}$ increases, the sensing OP decreases; also, increasing $c_1$ raises the sensing OP, indicating degraded sensing performance when more power is allocated to communication. Fig.~\ref{optarget}(a) demonstrates the tightness of the upper and lower bounds and the accuracy of the approximation $P_A$—the bounds become tighter as $c_1$ decreases. Fig.~\ref{optarget}(b) shows\footnote{For the numerical curves in Fig.~\ref{optarget}(b) and Fig.~\ref{optarget}(c), we plotted $P_A$, which closely matches the exact values as shown in Fig.~\ref{optarget}(a).} that the impact of the Rician factor $K$ on sensing performance is neither monotonic nor consistent. For $c_1=2.6$, the ordering of $K=0$, $0.32$, and $3.2$ changes with $\tilde{\epsilon}$: strong LoS performs best at very low thresholds, moderate LoS at slightly higher thresholds, and Rayleigh fading can outperform both at intermediate thresholds. For $c_1=3.0$, moderate and strong LoS are nearly identical across all $\tilde{\epsilon}$, both outperforming Rayleigh fading at low thresholds but being outperformed at mid-range thresholds. The coupling of sensing and communication through the shared power budget $p_t=|c_1|^2+|c_2|^2$ creates complex interactions that prevent a simple monotonic relationship between $K$ and sensing performance; the small gaps between curves indicate that $K$ plays a secondary role compared to power allocation and threshold selection.

Fig.~\ref{optarget}(c) illustrates the sensing OP of LB versus $\tilde{\epsilon}$ for different values of $(N,M)$ and $p_t$. The results indicate that increasing the number of antennas and $p_t$ improves sensing performance, with a greater impact observed from increasing the number of antennas compared to increasing $p_t$. Remarkably, the analytical OP closely match the Monte Carlo simulations even for as few as $N = 9$ antennas, confirming that the Gaussian approximation in (\ref{multyclt}) remains valid for OP calculation even with relatively small antenna arrays.
\begin{figure}
    \centering
    \includegraphics[scale=.46]{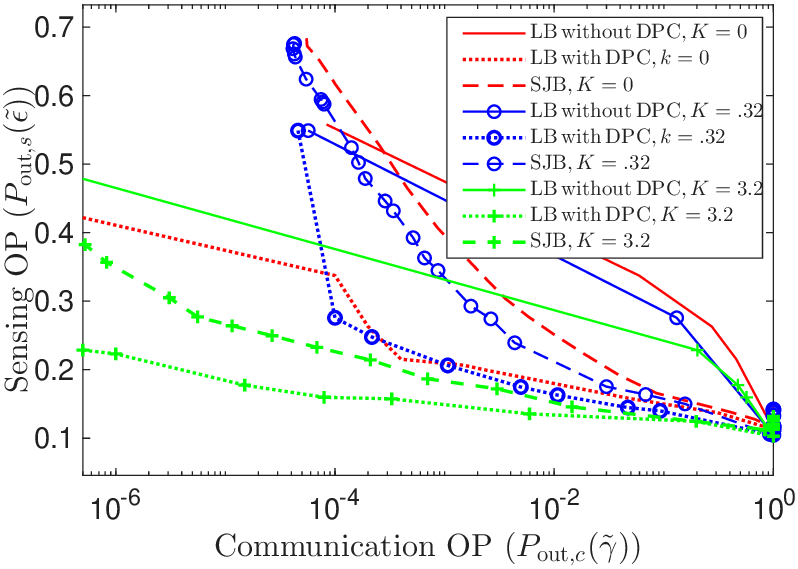}
    \caption{Sensing OP versus Communication OP in SJB and LB for $\tilde{\epsilon}=8 \times 10^{-7}$ and $\gamma=44$.}
    \label{region}
\end{figure}

Fig.~\ref{region} shows the achievable sensing-communication outage region for different Rician factors and beamforming schemes. The region is generated by varying the beamforming parameters ($|b_1|$ in SJB, $|c_1|$ in LB) for $N=15$ and different Rician factors $K$. Each point $(P_{\text{out},s}, P_{\text{out},c})$ represents a feasible operating point, with curves closer to the origin indicating better overall performance. Under the considered parameter setting, the figure reveals the achievable operating ranges of each scheme, including which configurations can achieve ultra-high communication reliability ($P_{\text{out},c} = 10^{-6}$), and how these limits depend on $K$ and the beamforming strategy. Several key observations emerge from the figure. 
\\
$\bullet$ For Rayleigh fading, only LB with DPC achieves ultra-high reliability down to $10^{-6}$; SJB and LB without DPC are limited to $P_{\text{out},c} \sim 10^{-4}$–$10^{-5}$. For moderate LoS, none of the schemes reach $10^{-6}$ under the considered parameter setting, with all curves stopping around the $10^{-4}$–$10^{-5}$ range, indicating an apparent reliability floor in this regime. For strong LoS, all schemes reach $10^{-6}$ by leveraging the deterministic path for coherent beamforming.
\\
$\bullet$ Examining the effect of increasing $K$ reveals distinct behaviors across the schemes. For SJB, increasing $K$ improves performance at moderate reliability, but only strong LoS enables ultra-high reliability. For LB with DPC, the behavior is non-monotonic: Rayleigh fading provides the full reliability range down to $10^{-6}$; moderate LoS improves moderate reliability but eliminates ultra-high reliability; while strong LoS delivers excellent performance across all reliability levels. For LB without DPC, improvement with $K$ is monotonic, but only strong LoS reaches $10^{-6}$.
\\
$\bullet$ The relative performance of the schemes also depends on the target reliability level. At moderate reliability ($10^{-3}$–$10^{-1}$), SJB and LB with DPC are comparable at low $K$, while LB with DPC dominates at high $K$; LB without DPC consistently lags. At high reliability ($10^{-5}$–$10^{-4}$), the behavior becomes more nuanced: for Rayleigh fading, only LB with DPC operates; for moderate LoS, all schemes hit the aforementioned floor; and for strong LoS, LB with DPC leads, followed by SJB and then LB without DPC. At ultra-high reliability ($10^{-6}$, only LB with DPC under Rayleigh fading and both LB with DPC and SJB under strong LoS are viable, with LB with DPC offering slightly better sensing performance.

The key insight is that, under the considered parameter setting, the moderate-LoS regime exhibits an apparent reliability floor around $10^{-4}$, likely due to phase mismatches or reduced diversity without sufficient coherent gain.
\section{Concluding Remarks}\label{conclusion}
This paper investigated the fundamental limits of a downlink MIMO ISAC system over Rician fading channels, including the Rayleigh special case $K=0$, where the BS--user channel consists of LoS and NLoS components and the user and target angles may follow an arbitrary joint distribution. For subspace joint beamforming (SJB) and linear beamforming (LB), we analytically characterized the communication outage probability (OP) and the sensing OP based on the Cram\'er--Rao bound (CRB), using tractable Gaussian formulations, simplified special-case expressions, and tractable bounds and approximations where closed-form evaluation is intractable. The results reveal a fundamental sensing--communication trade-off characterized by the Pareto-optimal outage region, and show that the Rician $K$-factor affects communication more strongly than sensing, with non-monotonic and regime-dependent behavior. From an asymptotic viewpoint, SJB exhibits channel hardening in communication and a sensing scaling of $\mathrm{CRB}=\mathcal{O}(1/N)$. LB with dirty paper coding (DPC) achieves linear communication array gain. For LB without DPC, the high-power regime is interference-limited due to radar self-interference, whereas in the large-array regime its communication behavior depends on angular alignment: for generic non-aligned user/target angle pairs, the communication outage probability vanishes as $N\to\infty$. For sensing, LB achieves the stronger scaling $\mathrm{CRB}=\Theta_p(1/N^2)$, outperforming SJB in the large-array regime. These analytical and numerical results also yield several design insights.
\\
1) For \textit{communication-centric operation}, LB with DPC is the strongest option and, under the considered settings, is the only scheme observed to achieve ultra-high reliability even in the Rayleigh case $K=0$. SJB offers a robust and predictable alternative, but typically requires strong LoS conditions to reach the same reliability range. LB without DPC is generally unattractive for communication-critical designs because radar self-interference severely degrades its outage performance. 
\\
2) For \textit{sensing-centric operation,} LB is generally preferable because it exploits more radar degrees of freedom than SJB and achieves the faster asymptotic sensing scaling $\Theta_p(1/N^2)$, making it particularly attractive for high-resolution massive MIMO sensing. This advantage is also consistent with the numerical results, with the largest gains observed in strong LoS environments where the deterministic component can be more effectively exploited. SJB, however, remains appealing when structural simplicity is important. In particular, SJB uses a single beamformer and a common waveform for sensing and communication, whereas LB with DPC requires separate beamformer design, power allocation between the two functionalities, and DPC processing, which makes implementation more demanding.
\\
3) For \textit{balanced ISAC designs}, the preferred scheme depends strongly on the propagation regime. In Rayleigh channels, LB with DPC provides the best communication reliability, while SJB remains competitive at moderate reliability levels. In the moderate-LoS regime, the numerical results under the considered parameter setting indicate an apparent communication reliability floor around $10^{-4}$, suggesting that this regime is unfavorable when ultra-high reliability is required. In strong-LoS environments, LB with DPC offers the best overall sensing--communication trade-off, while SJB provides a simpler and more robust alternative with somewhat lower peak performance. Overall, the results suggest that LB with DPC is the most powerful architecture when implementation complexity can be supported, especially in favorable LoS conditions, whereas SJB remains an attractive lower-complexity alternative that delivers robust performance across a broad range of operating conditions.
\appendices
%\section{Proof of Lemma \ref{lemma1i}}\label{lemma1pi}
%First, we derive ${{\boldsymbol{\mu}}_i} = [E[{{x}}_i], E[{{y}}_i], E[{{k}}_i]]$. Since \( |h_i|^2 \) is exponentially distributed with mean 1, we have $\mathbb{E}[|h_i|^2] = 1$. The expectation of each cosine term involving the random variable $\phi_i$ in the expression for $x_i$ in (\ref{xyk}) is equal to zero since $\phi_i$ is uniformly distributed. Therefore, noting that $\alpha^2_1+\alpha^2_2=1$, we have $\mathbb{E}[x_i] = |b_1| \cos(\phi_1)+|b_2| \alpha_1 \cos(\phi_2-f_i+f^u_i)$. Using a similar approach, we obtain $\mathbb{E}[y_i]$ and $\mathbb{E}[k_i]$. Next, we calculate each element of the covariance matrix, $\mathbf{{{\Sigma}_i}}$. Using $E(|h_i|^4) = 2$ and some standard trigonometric identities, we can evaluate the expectation of each term in (\ref{elemntofsigma}). The detailed derivations are omitted due to space limitations.
\section{Proof of Lemma \ref{lemmanewrician}}\label{lemmanewricianproof}
First, we analyze $C_1 = \sum_{i=1}^{N} \cos(f^u_i + \phi_2 - f_i)$. Let \(\Delta = \frac{\pi}{2}[\sin(\theta_u) - \sin(\theta)]\), then $f^u_i - f_i = \Delta[N - (2i-1)]$. It follows that $C_1 = \sum_{i=1}^{N} \cos(\phi_2 + \Delta[N - (2i-1)])
= \sum_{i=1}^{N} \cos(\phi_2 + \Delta N - \Delta(2i-1))$. Then, using \(\cos(a + b) = \cos a \cos b - \sin a \sin b\), we obtain $C_1 = \cos(\phi_2 + \Delta N)\sum_{i=1}^{N} \cos(\Delta(2i-1)) + \sin(\phi_2 + \Delta N)\sum_{i=1}^{N} \sin(\Delta(2i-1))$. So, it is required to find \(A = \sum_{i=1}^{N} \cos(\Delta(2i-1))\) and \(B = \sum_{i=1}^{N} \sin(\Delta(2i-1))\). For this, we use trigonometric identities for sums of cosines/sines in arithmetic progression, e.g., $\sum_{i=1}^{N} \cos(\alpha + (i-1)\beta) = \frac{\sin(\frac{N\beta}{2})}{\sin(\frac{\beta}{2})} \cos\left(\alpha + \frac{(N-1)\beta}{2}\right)$. For our case, by setting $\beta = 2\Delta$ and $\alpha = \Delta$ we obtain $A = \sum_{k=1}^{N} \cos(\Delta(2k-1)) = \frac{\sin(N\Delta)}{\sin(\Delta)} \cos(N\Delta)$ and $B = \sum_{k=1}^{N} \sin(\Delta(2k-1)) = \frac{\sin(N\Delta)}{\sin(\Delta)} \sin(N\Delta)$. Thus, $ C_1 = \cos(\phi_2 + \Delta N) \cdot \frac{\sin(N\Delta)}{\sin(\Delta)} \cos(N\Delta) + \sin(\phi_2 + \Delta N) \cdot \frac{\sin(N\Delta)}{\sin(\Delta)} \sin(N\Delta)$, which simplifies to $C_1 = \frac{\sin(N\Delta)}{\sin(\Delta)} \left[\cos(\phi_2 + \Delta N)\cos(N\Delta) + \sin(\phi_2 + \Delta N)\sin(N\Delta)\right]$. Thus, using \(\cos a\cos b + \sin a\sin b = \cos(a - b)\), we have $C_1 = \frac{\sin(N\Delta)}{\sin(\Delta)} \cos(\phi_2 + \Delta N - N\Delta) = \frac{\sin(N\Delta)}{\sin(\Delta)} \cos(\phi_2)$. A similar approach can then be applied to evaluate $S_1, C_2$ and $S_2$.
\section{Proof of Lemma \ref{lemma4}}\label{lemma4p}
To simplify the CRB, we begin by simplifying each term in the numerator and denominator of (\ref{crb}). Following a similar approach as in \cite[Appendix D, eq.(49)]{OnStochasticFundamentalLimitsina}, we have:
\begin{align}
&\text{Tr}(\mathbf{A}^H(\theta) \mathbf{A}(\theta) \mathbf{R}_x)=|c_1|^2||\mathbf{b}||^2|\mathbf{a}^H\mathbf{w}_1|^2+|c_2|^2||\mathbf{b}||^2|\mathbf{a}^H\mathbf{w}_2|^2.\label{denom1}
\end{align}
But from the definition of LB in Subsection \ref{linearbf}, \(|\mathbf{a}^H\mathbf{w}_2|^2 = N\). Also \(\mathbf{a}^H\mathbf{w}_1 = \frac{\mathbf{a}^H (\alpha_1\mathbf{a}(\theta_u)+\alpha_2\mathbf{h})}{\parallel (\alpha_1\mathbf{a}(\theta_u)+\alpha_2\mathbf{h}) \parallel}\). Let $\hat{R} =\sum_{i=1}^N \hat{r}_i,$ and $\hat{T}= \sum_{i=1}^N \hat{t}_i,$ $U\triangleq \sum_{i=1}^{N}u_i$, where $\hat{r}_i = \Re\left( \alpha_1 e^{j (f_i(\theta) - f_i(\theta_u))} + \alpha_2 e^{j f_i(\theta)} h_i \right),$ $\hat{t}_i = \Im\left( \alpha_1 e^{j (f_i(\theta) - f_i(\theta_u))} + \alpha_2 e^{j f_i(\theta)} h_i \right)$, and $u_i\triangleq |\alpha_1e^{-jf^u_i}+\alpha_2h_i|^2=\alpha^2_1+\alpha^2_2|h_i|^2+2\alpha_1\alpha_2|h_i|\cos(f^u_i+\phi_i)$. Then \(|\mathbf{a}^H\mathbf{w}_1|^2 = \frac{\hat{R}^2 + \hat{T}^2}{U}\). 
Thus, $\operatorname{Tr}(\mathbf{A}^H\mathbf{A}\mathbf{R}_x) = \|\mathbf{b}\|^2 \left( |c_1|^2 \frac{\hat{R}^2 + \hat{T}^2}{U} + |c_2|^2 N \right).$

%Moreover,  following a similar approach as in \cite[Appendix D,Eq.(50)]{OnStochasticFundamentalLimitsina}, we have
Moreover, after some algebraic simplifications, we obtain
\begin{align}
&\text{Tr}(\mathbf{A}'^H(\theta) \mathbf{A}'(\theta) \mathbf{R}_x)\overset{(a)}=|c_1|^2(||\mathbf{b}'||^2|\mathbf{a}^H\mathbf{w}_1|^2+||\mathbf{b}||^2|\mathbf{a}'^H\mathbf{w}_1|^2)\nonumber\\
&+|c_2|^2(||\mathbf{b}'||^2|\mathbf{a}^H\mathbf{w}_2|^2+||\mathbf{b}||^2|\mathbf{a}'^H\mathbf{w}_2|^2),\label{denom2}
\end{align}
But \({\mathbf{a}}'^H\mathbf{w}_2 = \frac{{\mathbf{a}}'^H\mathbf{a}}{\|\mathbf{a}\|} = 0\) (orthogonal). Also, we have $
{\mathbf{a}'}^H(\alpha_1\mathbf{a}(\theta_u)+\alpha_2\mathbf{h}) = j \sum_{i=1}^N f_i'(\theta) \left[ \alpha_1 e^{j (f_i(\theta) - f_i(\theta_u))} + \alpha_2 e^{j f_i(\theta)} h_i \right]= \sum_{i=1}^N (-f_i' \hat{t}_i) + j \sum_{i=1}^N (f_i' \hat{r}_i).$
Thus, $|{\mathbf{a}}'^H\mathbf{w}_1|^2 = \frac{ (\sum_i -f_i' \hat{t}_i)^2 + (\sum_i f_i' \hat{r}_i)^2 }{U}$ and
$\operatorname{Tr}({\mathbf{A}}'^H{\mathbf{A}}'\mathbf{R}_x) = \|{\mathbf{b}}'\|^2 \left( |c_1|^2 \frac{\hat{R}^2 + \hat{T}^2}{U} + |c_2|^2 N \right) + \|\mathbf{b}\|^2 |c_1|^2 \frac{ (\sum_i -f_i' \hat{t}_i)^2 + (\sum_i f_i' \hat{r}_i)^2 }{U}.$

%Moreover, following a similar approach as in \cite[Appendix D,Eq.(53)]{OnStochasticFundamentalLimitsina}, we have
Moreover, after some algebraic manipulations, we obtain
\begin{align}
&\text{Tr}(\mathbf{A}'^H(\theta) \mathbf{A}(\theta) \mathbf{R}_x)=|c_1|^2||\mathbf{b}||^2\mathbf{a}^H\mathbf{w}_1\mathbf{w}_1^H\mathbf{a}'\nonumber\\
&+|c_2|^2||\mathbf{b}||^2\mathbf{a}^H\mathbf{w}_2\mathbf{w}_2^H\mathbf{a}'\label{denom3}
\end{align}
The second term is zero due to the orthogonality of $\mathbf{a}$ and $\mathbf{a}'$. So
$|\operatorname{Tr}({\mathbf{A}}'^H\mathbf{A}\mathbf{R}_x)|^2 = \|\mathbf{b}\|^4 |c_1|^4 \frac{ (\hat{R}^2 + \hat{T}^2) \left[ (\sum_i -f_i' \hat{t}_i)^2 + (\sum_i f_i' \hat{r}_i)^2 \right] }{U^2}.$ Thus, the proof is complete.
\section{Proof of Proposition \ref{approximation}}\label{lemmafp}
When $N$ is large, we have
\begin{align}
&(\sum_{i=1}^{N}-f'_i\hat{t}_i)^2+(\sum_{i=1}^{N}f'_i\hat{r}_i)^2\approx E\{(\sum_{i=1}^{N}-f'_i\hat{t}_i)^2+(\sum_{i=1}^{N}f'_i\hat{r}_i)^2 \}
\end{align}
Let \(A = \sum_i f_i' \hat{r}_i\), \(B = \sum_i -f_i' \hat{t}_i\). Since the vectors ${\hat{\mathbf{d}}}_i$ are independent across \(i\) (but not identical), we have $E[A^2]= \sum_i (f_i')^2 \left( E[\hat{r}_i^2] - (E[\hat{r}_i])^2 \right) + \left( \sum_i f_i' E[\hat{r}_i] \right)^2=\sum_i (f_i')^2 \text{Var}(\hat{r}_i) + \left( \sum_i f_i' E[\hat{r}_i] \right)^2.$ Similarly $E[B^2] = \sum_i (f_i')^2 \text{Var}(\hat{t}_i) + \left( \sum_i -f_i' E[\hat{t}_i] \right)^2.$ Thus, due to Lemma \ref{lemma5i}: $E[A^2 + B^2] = \alpha_2^2 \sum_i (f_i')^2 + \alpha_1^2 \left[ \left( \sum_i f_i' \cos(f_i - f^u_i) \right)^2 + \left( \sum_i f_i' \sin(f_i - f^u_i) \right)^2 \right].$ Thus, we have
\begin{align}
&(\sum_{i=1}^{N}-f'_i\hat{t}_i)^2+(\sum_{i=1}^{N}f'_i\hat{r}_i)^2\approx
\alpha_2^2 \sum_{i=1}^N (f_i')^2 \nonumber\\&
+ \alpha_1^2 [( \sum_i f_i' \cos(f_i - f^u_i))^2 + ( \sum_i f_i' \sin(f_i - f^u_i))^2 ]\triangleq \tilde{A}.\label{lemmaf}
\end{align}
This gives a deterministic approximation of $(\sum_{i=1}^{N}-f'_i\hat{t}_i)^2+(\sum_{i=1}^{N}f'_i\hat{r}_i)^2$, making the CRB expression tractable. Thus, by replacing $(\sum_{i=1}^{N}-f'_i\hat{t}_i)^2+(\sum_{i=1}^{N}f'_i\hat{r}_i)^2\big)$ in the denominator of CRB($\theta$) in (\ref{crbsimplified}) with (\ref{lemmaf}), we obtain an approximation of CRB($\theta$), which we denote by ACRB. Thus, $P_{\text{out}, s}(\epsilon)\!=\!P(\text{CRB}(\theta)\!>\!\epsilon)\!\approx\!\! P(\text{ACRB}(\theta)\!>\!\epsilon)\!\!\triangleq\! P_{A}(\epsilon),$ where $P_{A}(\epsilon)$ is an approximation of the sensing OP. Then, following a similar approach as in Subsection \ref{opuserlb} and using the result of Lemma \ref{lemma5i}, the proof is complete.
\section{Proof of Lemma \ref{SJBuserNinf}}\label{ProofSJBuserNinf}
Define the average of the mean vector as $\bar{\boldsymbol{\mu}}_N = \frac{1}{N}\sum_{i=1}^N \boldsymbol{\mu}_i$. From Lemma 1, we have
$\sum_{i=1}^N \boldsymbol{\mu}_i = \begin{bmatrix}
N|b_1|\cos\phi_1 + |b_2|\alpha_1 C_1 \\
N|b_1|\sin\phi_1 + |b_2|\alpha_1 S_1 \\
N(|b_1|^2 + |b_2|^2) + 2\alpha_1|b_1b_2|C_2
\end{bmatrix}$.
For the generic case $\sin\Delta \neq 0$, the quantities $C_1$, $S_1$, and $C_2$ are uniformly bounded in $N$, since
$
\left|\frac{\sin(N\Delta)}{\sin(\Delta)}\right|\le \frac{1}{|\sin(\Delta)|}.
$
Therefore, $C_1/N$, $S_1/N$, and $C_2/N$ vanish as $N\to\infty$. Hence,
$\bar{\boldsymbol{\mu}}_N \to
\begin{bmatrix}
|b_1|\cos\phi_1\\
|b_1|\sin\phi_1\\
|b_1|^2+|b_2|^2
\end{bmatrix}$. For the variance, we have $\mathrm{Var}\left(\frac{1}{N}\sum_{i=1}^N x_i\right) = \frac{1}{N^2}\sum_{i=1}^N \mathrm{Var}(x_i) \leq \frac{1}{N^2} \cdot N \cdot \max_i \mathrm{Var}(x_i) = \frac{C}{N} \to 0.$ By Chebyshev's inequality, for any $\epsilon > 0$, $P\left(\left|\frac{1}{N}\sum_{i=1}^N x_i - \frac{1}{N}\sum_{i=1}^N \mu_{x,i}\right| > \epsilon\right) \leq \frac{\mathrm{Var}(\frac{1}{N}\sum x_i)}{\epsilon^2} \leq \frac{C}{N\epsilon^2} \to 0.$ Thus $\frac{1}{N}\sum_{i=1}^N x_i - \frac{1}{N}\sum_{i=1}^N \mu_{x,i} \xrightarrow{p} 0$. Since $\frac{1}{N}\sum_{i=1}^N \mu_{x,i} \to \mu_{x,\infty}$, we have $\frac{X}{N} \xrightarrow{p} \mu_{x,\infty}$. The same holds for $Y$ and $Z$. For almost sure convergence, we invoke Kolmogorov's strong law of large numbers for independent non-identically distributed random variables. The variables $x_i$ are independent with $\mathbb{E}[x_i] = \mu_{x,i}$ and $\mathrm{Var}(x_i) \leq C$. Kolmogorov's criterion requires $\sum_{i=1}^\infty \frac{\mathrm{Var}(x_i)}{i^2} < \infty$. Since $\mathrm{Var}(x_i) \leq C$, we have $\sum_{i=1}^\infty \frac{C}{i^2} < \infty$, so the condition holds. Therefore, $\frac{1}{N}\sum_{i=1}^N (x_i - \mu_{x,i}) \xrightarrow{\text{a.s.}} 0.$ Together with $\frac{1}{N}\sum_{i=1}^N \mu_{x,i} \to \mu_{x,\infty}$, this yields $\frac{X}{N} \xrightarrow{\text{a.s.}} \mu_{x,\infty}$. The same argument applies to $Y$ and $Z$. Finally, from (\ref{sinr2i}), $\mathrm{SINR} = \frac{p_t c^2}{\sigma_u^2} \cdot \frac{(X/N)^2 + (Y/N)^2}{Z/N} \xrightarrow{\text{a.s.}} \frac{p_t c^2}{\sigma_u^2} \cdot \frac{\mu_{x,\infty}^2 + \mu_{y,\infty}^2}{\mu_{k,\infty}} \triangleq \mathrm{SINR}_\infty^{\text{(SJB)}}.$ We now explicitly compute the limits \(\mu_{x,\infty}, \mu_{y,\infty}, \mu_{k,\infty}\). Using Lemma \ref{lemma1i} and Lemma \ref{lemmanewrician}, for the generic case where \(\sin\Delta \neq 0\), the term \(\frac{\sin(N\Delta)}{\sin\Delta}\) is bounded. Consequently, $\lim_{N\to\infty} \frac{1}{N}\sum_{i=1}^N \mu_{x,i} = |b_1|\cos\phi_1 + |b_2|\alpha_1 \underbrace{\lim_{N\to\infty}\frac{1}{N}\cdot\frac{\sin(N\Delta)}{\sin\Delta}}_{=0}\cos\phi_2 = |b_1|\cos\phi_1.$ Following the same reasoning, we obtain $\mu_{y,\infty} = |b_1|\sin\phi_1, \qquad \mu_{k,\infty} = |b_1|^2 + |b_2|^2.$ Thus the proof is complete.
\bibliographystyle{IEEEtran}
\bibliography{ref}

@ARTICLE{Integratedtoward,
  author={Liu, Fan and Cui, Yuanhao and Masouros, Christos and Xu, Jie and Han, Tony Xiao and Eldar, Yonina C. and Buzzi, Stefano},
  journal={IEEE J. Sel. Areas Commun.}, 
  title={Integrated Sensing and Communications: Toward Dual-Functional Wireless Networks for 6G and Beyond}, 
  year={2022},
  volume={40},
  number={6},
  pages={1728-1767},
  doi={10.1109/JSAC.2022.3156632}}

@ARTICLE{PerformanceAnalysisofISACWithActiveMulti,
  author={Singh Parihar, Abhinav and Singh, Keshav and Alexandropoulos, George C. and Ding, Zhiguo and Li, Chih-Peng},
  journal={IEEE Transactions on Cognitive Communications and Networking}, 
  title={Performance Analysis of ISAC With Active Multi-Functional RISs in Rician Fading Channels}, 
  year={2026},
  volume={12},
  number={},
  pages={355-369},
  doi={10.1109/TCCN.2025.3561312}}

@ARTICLE{costa,
  author={Costa, M.},
  journal={IEEE Trans. Inf. Theory}, 
  title={Writing on dirty paper (Corresp.)}, 
  year={1983},
  volume={29},
  number={3},
  pages={439-441},
  doi={10.1109/TIT.1983.1056659}}

@ARTICLE{TargetDetectionandLocalization,
  author={Bekkerman, I. and Tabrikian, J.},
  journal={IEEE Trans. Signal Process.}, 
  title={Target Detection and Localization Using {MIMO} Radars and Sonars}, 
  year={2006},
  volume={54},
  number={10},
  pages={3873-3883},
  doi={10.1109/TSP.2006.879267}}

@ARTICLE{optimaltransmitbeamformingintegrated,
  author={Hua, Haocheng and Xu, Jie and Han, Tony Xiao},
  journal={IEEE Trans. Veh. Technol.}, 
  title={Optimal Transmit Beamforming for Integrated Sensing and Communication}, 
  year={2023},
  volume={72},
  number={8},
  pages={10588-10603},
  doi={10.1109/TVT.2023.3262513}}

@ARTICLE{networklevelintegratedsensingcommunication,
  author={Meng, Kaitao and Masouros, Christos and Chen, Guangji and Liu, Fan},
  journal={IEEE Trans. Wireless Commun.}, 
  title={Network-Level Integrated Sensing and Communication: Interference Management and {BS} Coordination Using Stochastic Geometry}, 
  year={2024},
  volume={23},
  number={12},
  pages={19365-19381},
  doi={10.1109/TWC.2024.3483031}}

@ARTICLE{coverageandrateofjointcommunication,
  author={Olson, Nicholas R. and Andrews, Jeffrey G. and Heath, Robert W.},
  journal={IEEE Trans. Inf. Theory}, 
  title={Coverage and Rate of Joint Communication and Parameter Estimation in Wireless Networks}, 
  year={2024},
  volume={70},
  number={1},
  pages={206-243},
  doi={10.1109/TIT.2023.3334200}}

@ARTICLE{fundamentalcrbratetradeoffmultiantenna,
  author={Ren, Zixiang and Peng, Yunfei and Song, Xianxin and Fang, Yuan and Qiu, Ling and Liu, Liang and Ng, Derrick Wing Kwan and Xu, Jie},
  journal={IEEE Trans. Wireless Commun.}, 
  title={Fundamental {CRB}-Rate Tradeoff in Multi-Antenna {ISAC} Systems With Information Multicasting and Multi-Target Sensing}, 
  year={2024},
  volume={23},
  number={4},
  pages={3870-3885},
  doi={10.1109/TWC.2023.3312723}}

@ARTICLE{MUMIMOCommunicationsWithMIMORadar,
  author={Liu, Fan and Masouros, Christos and Li, Ang and Sun, Huafei and Hanzo, Lajos},
  journal={IEEE Trans. Wireless Commun.}, 
  title={{MU-MIMO} Communications With {MIMO} Radar: From Co-Existence to Joint Transmission}, 
  year={2018},
  volume={17},
  number={4},
  pages={2755-2770},
  doi={10.1109/TWC.2018.2803045}}

@ARTICLE{TowardDualfunctionalRadarCommunicationSystems,
  author={Liu, Fan and Zhou, Longfei and Masouros, Christos and Li, Ang and Luo, Wu and Petropulu, Athina},
  journal={IEEE Trans. Signal Process.}, 
  title={Toward Dual-functional Radar-Communication Systems: Optimal Waveform Design}, 
  year={2018},
  volume={66},
  number={16},
  pages={4264-4279},
  doi={10.1109/TSP.2018.2847648}}

@misc{MIMOIntegratedSensingandCommunicationCRBRateTradeoff,
      title={{MIMO} Integrated Sensing and Communication: {CRB}-Rate Tradeoff}, 
      author={Haocheng Hua and Tony Xiao Han and Jie Xu},
      year={2022},
      eprint={2209.12721},
      archivePrefix={arXiv},
      primaryClass={cs.IT},
      url={https://arxiv.org/abs/2209.12721}, 
}

@ARTICLE{CrameRaoBoundOptimizationforJoint,
  author={Liu, Fan and Liu, Ya-Feng and Li, Ang and Masouros, Christos and Eldar, Yonina C.},
  journal={IEEE Trans. Signal Process.}, 
  title={Cramér-Rao Bound Optimization for Joint Radar-Communication Beamforming}, 
  year={2022},
  volume={70},
  number={},
  pages={240-253},
  doi={10.1109/TSP.2021.3135692}}

@ARTICLE{SeventyYearsofRadarandCommunications,
  author={Liu, Fan and Zheng, Le and Cui, Yuanhao and Masouros, Christos and Petropulu, Athina P. and Griffiths, Hugh and Eldar, Yonina C.},
  journal={IEEE Signal Process. Magazine}, 
  title={Seventy Years of Radar and Communications: The road from separation to integration}, 
  year={2023},
  volume={40},
  number={5},
  pages={106-121},
  doi={10.1109/MSP.2023.3272881}}

@ARTICLE{ASurveyonFundamentalLimits,
  author={Liu, An and Huang, Zhe and Li, Min and Wan, Yubo and Li, Wenrui and Han, Tony Xiao and Liu, Chenchen and Du, Rui and Tan, Danny Kai Pin and Lu, Jianmin and Shen, Yuan and Colone, Fabiola and Chetty, Kevin},
  journal={IEEE Commun. Surv. Tutor.}, 
  title={A Survey on Fundamental Limits of Integrated Sensing and Communication}, 
  year={2022},
  volume={24},
  number={2},
  pages={994-1034},
  doi={10.1109/COMST.2022.3149272}}

@INPROCEEDINGS{OnthePerformanceofUplinkandDownlink,
  author={Liu, Meng and Yang, Minglei and Nallanathan, Arumugam},
  booktitle={Proc. 2022 IEEE Globecom Workshops (GC Wkshps)}, 
  title={On the Performance of Uplink and Downlink Integrated Sensing and Communication Systems}, 
  year={2022},
  volume={},
  number={},
  pages={1236-1241},
  doi={10.1109/GCWkshps56602.2022.10008554}}

@ARTICLE{PerformanceAnalysisandPowerAllocationforCooperative,
  author={Liu, Meng and Yang, Minglei and Li, Huifang and Zeng, Kun and Zhang, Zhaoming and Nallanathan, Arumugam and Wang, Guangjian and Hanzo, Lajos},
  journal={IEEE Internet Things J.}, 
  title={Performance Analysis and Power Allocation for Cooperative {ISAC} Networks}, 
  year={2023},
  volume={10},
  number={7},
  pages={6336-6351},
  doi={10.1109/JIOT.2022.3225281}}

@book{elgamal, place={Cambridge}, title={Network Information Theory}, publisher={Cambridge University Press}, author={El Gamal, Abbas and Kim, Young-Han}, year={2011}}

@misc{NOMAISACPerformanceAnalysisandRateRegion,
      title={{NOMA}-{ISAC}: Performance Analysis and Rate Region Characterization}, 
      author={Chongjun Ouyang and Yuanwei Liu and Hongwen Yang},
      year={2022},
      eprint={2205.13756},
      archivePrefix={arXiv},
      primaryClass={cs.IT},
      url={https://arxiv.org/abs/2205.13756}, 
}

@ARTICLE{Aunifiedperformanceframeworkfor,
  author={Al-Jarrah, Mohammad and Alsusa, Emad and Masouros, Christos},
  journal={IEEE Trans. Wireless Commun.}, 
  title={A Unified Performance Framework for Integrated Sensing-Communications Based on KL-Divergence}, 
  year={2023},
  volume={22},
  number={12},
  pages={9390-9411},
  doi={10.1109/TWC.2023.3270390}}

@ARTICLE{PerformanceofDownlinkandUplinkIntegratedSensing,
  author={Ouyang, Chongjun and Liu, Yuanwei and Yang, Hongwen},
  journal={IEEE Wireless Commun. Lett.}, 
  title={Performance of Downlink and Uplink Integrated Sensing and Communications ({ISAC}) Systems}, 
  year={2022},
  volume={11},
  number={9},
  pages={1850-1854},
  doi={10.1109/LWC.2022.3184409}}

@INPROCEEDINGS{JointTransmitBeamformingforMultiuser,
  author={Liu, Xiang and Huang, Tianyao and Liu, Yimin and Zhou, Jie},
  booktitle={Proc. 2019 IEEE Int. Conf. Signal, Inf. Data Process. (ICSIDP)}, 
  title={Joint Transmit Beamforming for Multiuser {MIMO} Communication and {MIMO} Radar}, 
  year={2019},
  volume={},
  number={},
  pages={1-6},
  doi={10.1109/ICSIDP47821.2019.9173030}}

@INPROCEEDINGS{PerformanceAnalysioftheFullDuplexJoint,
  author={Guo, Yinghong and Li, Cheng and Zhang, Chaoxian and Yao, Yao and Xia, Bin},
  booktitle={Proc. 2021 IEEE/CIC Int. Conf. Commun. China (ICCC)}, 
  title={Performance Analysis of the Full-Duplex Joint Radar and Communication System}, 
  year={2021},
  volume={},
  number={},
  pages={505-510},
  doi={10.1109/ICCC52777.2021.9580217}}

@misc{MIMOISACPerformanceAnalysis,
      title={{MIMO}-{ISAC}: Performance Analysis and Rate Region Characterization}, 
      author={Chongjun Ouyang and Yuanwei Liu and Hongwen Yang},
      year={2023},
      eprint={2209.01028},
      archivePrefix={arXiv},
      primaryClass={cs.IT},
      url={https://arxiv.org/abs/2209.01028}, 
}

@article{Amethodtointegrate,
    author = {Das, Abhranil and Geisler, Wilson S.},
    title = {A method to integrate and classify normal distributions},
    journal = {Journal of Vision},
    volume = {21},
    number = {10},
    pages = {1-1},
    year = {2021},
    month = {09},
    issn = {1534-7362},
    doi = {10.1167/jov.21.10.1},
    url = {https://doi.org/10.1167/jov.21.10.1},
    eprint = {https://arvojournals.org/arvo/content\_public/journal/jov/938556/i1534-7362-21-10-1\_1630475357.9069.pdf},
}

@INPROCEEDINGS{2018JointStateSensingandCommunicationOptimalTradeoffforaMemorylessCase,
  author={Kobayashi, Mari and Caire, Giuseppe and Kramer, Gerhard},
  booktitle={Proc. 2018 IEEE Int. Symp. Inf. Theory (ISIT)}, 
  title={Joint State Sensing and Communication: Optimal Tradeoff for a Memoryless Case}, 
  year={2018},
  volume={},
  number={},
  pages={111-115},
  doi={10.1109/ISIT.2018.8437621}}

@ARTICLE{2022AnInformation-TheoreticApproachtoJointSensingandCommunication,
  author={Ahmadipour, Mehrasa and Kobayashi, Mari and Wigger, Michele and Caire, Giuseppe},
  journal={IEEE Trans. Inf. Theory}, 
  title={An Information-Theoretic Approach to Joint Sensing and Communication}, 
  year={2024},
  volume={70},
  number={2},
  pages={1124-1146},
  doi={10.1109/TIT.2022.3176139}}

@ARTICLE{onthefundementaltradeoff,
  author={Xiong, Yifeng and Liu, Fan and Cui, Yuanhao and Yuan, Weijie and Han, Tony Xiao and Caire, Giuseppe},
  journal={IEEE Trans. Inf. Theory}, 
  title={On the Fundamental Tradeoff of Integrated Sensing and Communications Under {G}aussian Channels}, 
  year={2023},
  volume={69},
  number={9},
  pages={5723-5751},
  doi={10.1109/TIT.2023.3284449}}

@article{fromtorchtoprojector,
author = {Xiong, Yifeng and Liu, Fan and Wan, Kai and Yuan, Weijie and Yuanhao, Cui and Caire, Giuseppe},
year = {2024},
month = {01},
pages = {1-13},
title = {From Torch to Projector: Fundamental Tradeoff of Integrated Sensing and Communications},
volume = {PP},
journal = {IEEE BITS Inf. Theory Mag.},
doi = {10.1109/MBITS.2024.3376638}
}

@INPROCEEDINGS{soltani2023outage,
  author={Soltani, Marziyeh and Mirmohseni, Mahtab and Tafazolli, Rahim},
  booktitle={2023 IEEE Globecom Workshops (GC Wkshps)}, 
  title={Outage Tradeoff Analysis in a Downlink Integrated Sensing and Communication Network}, 
  year={2023},
  volume={},
  number={},
  pages={951-956},
  doi={10.1109/GCWkshps58843.2023.10464786}}

@ARTICLE{OnStochasticFundamentalLimitsina,
  author={Soltani, Marziyeh and Mirmohseni, Mahtab and Tafazolli, Rahim},
  journal={IEEE Transactions on Communications}, 
  title={On Stochastic Fundamental Limits in a Downlink Integrated Sensing and Communication Network}, 
  year={2025},
  volume={},
  number={},
  pages={1-1},
  doi={10.1109/TCOMM.2025.3578814}}

@article{Fundamentalchannelcoupling,
author = {Gan, Xu and Chongwen, Huang and Yang, Zhaohui and Chen, Xiaoming and Liu, Fan and Zhang, Zhaoyang and Yuen, Chau and Guan, Yong and Debbah, mérouane},
year = {2025},
month = {07},
pages = {},
title = {Fundamental channel coupling effects for integrated sensing and communication systems},
volume = {68},
journal = {Science China Information Sciences},
doi = {10.1007/s11432-024-4417-y}
}

@ARTICLE{Network-levelISAC:AnAnalyticalStudyofAntenna,
  author={Meng, Kaitao and Han, Kawon and Masouros, Christos and Hanzo, Lajos},
  journal={IEEE Transactions on Wireless Communications}, 
  title={Network-level ISAC: An Analytical Study of Antenna Topologies Ranging from Massive to Cell-Free MIMO}, 
  year={2025},
  volume={},
  number={},
  pages={1-1},
  doi={10.1109/TWC.2025.3576432}}

@misc{PerformanceAnalysisofCooperativeIntegratedSensing,
      title={Performance Analysis of Cooperative Integrated Sensing and Communications for 6G Networks}, 
      author={Dongsheng Sui and Cunhua Pan and Hong Ren and Jiahua Wan and Liuchang Zhuo and Jing Jin and Qixing Wang and Jiangzhou Wang},
      year={2025},
      eprint={2505.08221},
      archivePrefix={arXiv},
      primaryClass={eess.SP},
      url={https://arxiv.org/abs/2505.08221}, 
}

@misc{OnStochasticPerformanceAnalysisofSecureIntegratedSensingandCommunicationNetworks,
      title={On Stochastic Performance Analysis of Secure Integrated Sensing and Communication Networks}, 
      author={Marziyeh Soltani and Mahtab Mirmohseni and Rahim Tafazolli},
      year={2025},
      eprint={2504.09674},
      archivePrefix={arXiv},
      primaryClass={cs.IT},
      url={https://arxiv.org/abs/2504.09674}, 
}

@ARTICLE{PhysicalLayerSecurityOptimizationWithCramér–RaoBoundMetric,
  author={Jia, Hanbo and Li, Xiaoshuai and Ma, Lin},
  journal={IEEE Trans. Veh. Technol.}, 
  title={Physical Layer Security Optimization With Cramér–Rao Bound Metric in {ISAC} Systems Under Sensing-Specific Imperfect {CSI} Model}, 
  year={2024},
  volume={73},
  number={5},
  doi={10.1109/TVT.2023.3347527}}
\end{document}